\documentclass[10pt,a4paper]{amsart}
\usepackage[margin=20mm]{geometry}
\usepackage[numbers]{natbib}
\usepackage[colorlinks]{hyperref}

\usepackage{amssymb,amsmath}
\usepackage{amsthm}
\usepackage{bbm}
\usepackage{stmaryrd}
\usepackage{hyperref}
\usepackage[usenames,dvipsnames]{xcolor}
\usepackage{color}

\input xy
\xyoption{all} 

\newcommand{\cal}{\mathcal}

\numberwithin{equation}{section}

\newtheorem{theorem}{Theorem}[section]
\newtheorem{corollary}[theorem]{Corollary}

\newtheorem{proposition}[theorem]{Proposition}

\theoremstyle{definition}
\newtheorem{definition}[theorem]{Definition}
\newtheorem{remark}[theorem]{Remark}
\newtheorem{example}[theorem]{Example}
\newtheorem{coexample}[theorem]{Counter Example}

\newcommand{\Dleft}{[\hspace{-1.5pt}[}
\newcommand{\Dright}{]\hspace{-1.5pt}]}
\newcommand{\SN}[1]{\Dleft #1 \Dright}

\newcommand{\rmd}{\textnormal{d}}

\newcommand{\rmh}{\textnormal{h}}

\newcommand{\rml}{\textnormal{l}}

\DeclareMathOperator{\w}{w}

\font\black=cmbx10 \font\sblack=cmbx7 \font\ssblack=cmbx5 \font\blackital=cmmib10  \skewchar\blackital='177
\font\sblackital=cmmib7 \skewchar\sblackital='177 \font\ssblackital=cmmib5 \skewchar\ssblackital='177
\font\sanss=cmss10 \font\ssanss=cmss8 
\font\sssanss=cmss8 scaled 600 \font\blackboard=msbm10 \font\sblackboard=msbm7 \font\ssblackboard=msbm5
\font\caligr=eusm10 \font\scaligr=eusm7 \font\sscaligr=eusm5  \font\fraktur=eufm10
\font\sfraktur=eufm7 \font\ssfraktur=eufm5 
\font\bsymb=cmsy10 scaled\magstep2
\def\all#1{\setbox0=\hbox{\lower1.5pt\hbox{\bsymb
       \char"38}}\setbox1=\hbox{$_{#1}$} \box0\lower2pt\box1\;}
\def\exi#1{\setbox0=\hbox{\lower1.5pt\hbox{\bsymb \char"39}}
       \setbox1=\hbox{$_{#1}$} \box0\lower2pt\box1\;}

\def\tx#1{{\fam0\relax#1}}

\newfam\bifam
\textfont\bifam=\blackital \scriptfont\bifam=\sblackital \scriptscriptfont\bifam=\ssblackital

\newfam\blfam
\textfont\blfam=\black \scriptfont\blfam=\sblack \scriptscriptfont\blfam=\ssblack

\newfam\bbfam
\textfont\bbfam=\blackboard \scriptfont\bbfam=\sblackboard \scriptscriptfont\bbfam=\ssblackboard

\newfam\ssfam
\textfont\ssfam=\sanss \scriptfont\ssfam=\ssanss \scriptscriptfont\ssfam=\sssanss
\def\sss#1{{\fam\ssfam\relax#1}}

\newfam\clfam
\textfont\clfam=\caligr \scriptfont\clfam=\scaligr \scriptscriptfont\clfam=\sscaligr

\newfam\frfam
\textfont\frfam=\fraktur \scriptfont\frfam=\sfraktur \scriptscriptfont\frfam=\ssfraktur

\def\hpb#1{\setbox0=\hbox{${#1}$}
    \copy0 \kern-\wd0 \kern.2pt \box0}
\def\vpb#1{\setbox0=\hbox{${#1}$}
    \copy0 \kern-\wd0 \raise.08pt \box0}

\def\pmb#1{\setbox0\hbox{${#1}$} \copy0 \kern-\wd0 \kern.2pt \box0}
\def\pmbb#1{\setbox0\hbox{${#1}$} \copy0 \kern-\wd0
      \kern.2pt \copy0 \kern-\wd0 \kern.2pt \box0}
\def\pmbbb#1{\setbox0\hbox{${#1}$} \copy0 \kern-\wd0
      \kern.2pt \copy0 \kern-\wd0 \kern.2pt
    \copy0 \kern-\wd0 \kern.2pt \box0}
\def\pmxb#1{\setbox0\hbox{${#1}$} \copy0 \kern-\wd0
      \kern.2pt \copy0 \kern-\wd0 \kern.2pt
      \copy0 \kern-\wd0 \kern.2pt \copy0 \kern-\wd0 \kern.2pt \box0}
\def\pmxbb#1{\setbox0\hbox{${#1}$} \copy0 \kern-\wd0 \kern.2pt
      \copy0 \kern-\wd0 \kern.2pt
      \copy0 \kern-\wd0 \kern.2pt \copy0 \kern-\wd0 \kern.2pt
      \copy0 \kern-\wd0 \kern.2pt \box0}

\mathchardef\za="710B  
\mathchardef\zb="710C  
\mathchardef\zg="710D  
\mathchardef\zd="710E  
\mathchardef\zve="710F 
\mathchardef\zz="7110  
\mathchardef\zh="7111  
\mathchardef\zvy="7112 
\mathchardef\zi="7113  
\mathchardef\zk="7114  
\mathchardef\zl="7115  
\mathchardef\zm="7116  
\mathchardef\zn="7117  
\mathchardef\zx="7118  
\mathchardef\zp="7119  
\mathchardef\zr="711A  
\mathchardef\zs="711B  
\mathchardef\zt="711C  
\mathchardef\zu="711D  
\mathchardef\zvf="711E 
\mathchardef\zq="711F  
\mathchardef\zc="7120  
\mathchardef\zw="7121  
\mathchardef\ze="7122  
\mathchardef\zy="7123  
\mathchardef\zf="7124  
\mathchardef\zvr="7125 
\mathchardef\zvs="7126 
\mathchardef\zf="7127  
\mathchardef\zG="7000  
\mathchardef\zD="7001  
\mathchardef\zY="7002  
\mathchardef\zL="7003  
\mathchardef\zX="7004  
\mathchardef\zP="7005  
\mathchardef\zS="7006  
\mathchardef\zU="7007  
\mathchardef\zF="7008  
\mathchardef\zW="700A  
\mathchardef\zC="7009  

\newcommand{\be}{\begin{equation}}
\newcommand{\ee}{\end{equation}}

\newcommand{\bea}{\begin{eqnarray}}
\newcommand{\eea}{\end{eqnarray}}
\def\*{{\textstyle *}}
\newcommand{\R}{{\mathbb R}}

\newcommand{\s}{{\textstyle *}}

\newcommand{\ti}{\times}
\newcommand{\A}{{\cal A}}

\def\Sec{\sss{Sec}}
\def\Vect{\sss{Vect}}

\def\sT{{\sss T}}

\def\xi{\tx{i}}

\def\xd{\operatorname{d}}

\def\s*{{\scriptstyle *}}

\def\cO{\mathcal{O}}

\def\cL{\mathcal{L}}

\newcommand{\beas}{\begin{eqnarray*}}
\newcommand{\eeas}{\end{eqnarray*}}

\def\half{\frac{1}{2}}

\begin{document}
\bibliographystyle{plain}

\author{Andrew James Bruce}
        \address{Mathematics Research Unit, University of Luxembourg, Maison du Nombre 6, avenue de la Fonte,
L-4364 Esch-sur-Alzette, LUXEMBOURG.}\email{andrewjamesbruce@googlemail.com}

\author{Janusz Grabowski}
       \address{Institute of Mathematics, Polish Academy of Sciences, ul. \'{S}niadeckich 8, 00--656 Warsaw, POLAND.}\email{jagrab@impan.pl}

\date{\today}
\title{Pre-Courant algebroids}

 \thanks{The research of JG was funded by the  Polish National Science Centre grant under the contract number DEC-2012/06/A/ST1/00256. }

\begin{abstract}
Pre-Courant algebroids are  `Courant algebroids' without the Jacobi identity for the Courant--Dorfman bracket.  We examine the corresponding supermanifold description of pre-Courant algebroids and some direct consequences thereof. In particular, we define \emph{symplectic almost Lie 2-algebroids} and show how they correspond to pre-Courant algebroids.  We give the definition of (sub-)Dirac structures and study  the na\"{\i}ve quasi-cochain complex within the setting of supergeometry. Moreover, the framework of supermanifolds allows us to economically define and work with pre-Courant algebroids equipped with a compatible non-negative grading. VB-Courant algebroids are  natural examples of  what we  call \emph{weighted pre-Courant algebroids} and our approach drastically simplifies working with them.
\end{abstract}

\maketitle
\vspace{-120pt}
\begin{center}
\emph{Dedicated to the memory of James Alfred Bruce}
\end{center}
\vspace{100pt}

\begin{small}
\smallskip

\end{small}

\tableofcontents

\section{Introduction}
There has been a recent drive in the `higher categorification' of mathematics and physics. The categorical approach to `higher structures' is extremely powerful and has given rise to multisymplectic manifolds, Courant algebroids, Dirac submanifolds and $L_{\infty}$-algebras,  to name just  a few mathematical structures. From a physics perspective, the categorical way of thinking has proved useful in axiomatising  the path integral approach to conformal and topological field theories, for example.  However, this categorical approach to `higher structures' can be difficult to work with and this is especially so in the context of finding concrete examples and applications. Allied to this categorical perspective is the study of supermanifolds equipped with some geometrical structure and a compatible additional grading  (see for example \cite{Grabowski:2009, Grabowski:2012, Roytenberg:2002, Severa:2005,Voronov:2001}). In particular, \emph{Q-manifolds} (see \cite{Alexandrov:1997}) with additional gradings on their structure sheaves offers a framework that is much easier to handle than `higher categorification'. For example, Lie algebroids are economically described in these terms following Va\u{\i}ntrob \cite{Vaintrob:1997}.  We remind the reader that a Q-manifold is  a supermanifold $M$, equipped with a Grassmann odd vector field, typically denoted by $Q$, that `squares to zero', i.e., $ Q^2 = \half[Q,Q] =0$.   Note that this condition is non-trivial due to the  $\mathbb{Z}_2$-grading. The vector field $Q \in \Vect(M)$ is typically referred to as a \emph{homological vector field}.   \par
A  \emph{classical gauge system} is a  supermanifold $M$, equipped with an even/odd Poisson bracket and an odd/even \emph{homological potential}  $\theta \in C^\infty(M)$. The further requirement is that the \emph{classical master equation}, $\{\theta,\theta \}=0$, holds.  It is clear that the Hamiltonian vector field associated with the homological potential
$$Q := \{\theta,  \: -  \: \},$$
is, in fact, a homological vector field.  The language here is of course borrowed from the BV-BRST and BFV formalisms of gauge theory.   Many  mathematically interesting and physically relevant   Q-manifolds arise from classical gauge systems, see Lyakhovich \&  Sharapov \cite{Lyakhovich:2004}. For instance,  Courant algebroids can be understood as classical gauge systems \`{a} la Roytenberg \cite{Roytenberg:2002, Roytenberg:2009} and \v{S}evera \cite{Severa:2017}. \par
Loosely, a Courant algebroid (see \cite{Courant:1990,Dorfman:1987,Liu:1997}) is a vector bundle whose space of sections comes equipped with a Loday--Leibniz algebra structure, together with an anchor map and a nondegenerate symmetric bilinear form that satisfy  some compatibility conditions. In particular, the bracket on the space of sections is \emph{not} a Lie bracket, but rather a non-skewsymmetric bracket that satisfies a version of the Jacobi identity.  For a recent review of the history of Courant algebroids, the reader should consult Kosmann-Schwarzbach \cite{Kosmann-Schwarzbach:2013}. \par
 The r\^{o}le of Courant algebroids in physics is two-fold. Firstly,  Courant algebroids were originally developed in order to study the integrability of Dirac structures. Thus,  Courant algebroids are important in geometric mechanics in the presence of constraints. Secondly, Courant algebroids appear to be fundamental in the supergravity limit of string theory, see for example \cite{Asakawa:2015,Bessho:2016}.  The observation that Courant algebroids are related to two-dimensional variational problems is due to \v{S}evera \cite{Severa:2005,Severa:2017}. Courant algebroids also  appear in the context of closed string theory with a  toroidal target space via double field theory,  see for example Deser \& Stasheff \cite{Deser:2015}.  Another connection with string theory is the fact that D-branes  can be identified with certain Dirac structures, see Asakawa, Sasa \& Watamura \cite{Asakawa:2012}. \par
A  \emph{pre-Courant algebroid} is a `Courant algebroid' for which the Jacobi identity for the bracket has been discarded (see Vaisman \cite{Vaisman:2005}). Pre-Courant algebroids naturally appear in the theory of three-dimensional sigma models with Wess-Zumino terms, see Hansen \& Strobl \cite{Hansen:2010}. Furthermore,  pre-Courant algebroids  appear in parabolic geometry via the work of Armstrong \cite{Armstrong:2007} and Armstrong \& Lu \cite{Armstrong:2011}. They are also to be found in the context of Cartan geometry via the work of Xu \cite{Xu:2014}. In short, Courant and  pre-Courant algebroids appear in a range of contexts spanning both pure mathematics and theoretical physics. \par
In this paper, we  reformulate the notion of a pre-Courant algebroid using the language of supermanifolds as suggested in \cite{Liu:2016}. By relaxing the condition that the  function $\theta$ be a homological potential, we arrive at a `weakened' version of a Courant algebroid. That is, we do not assume that the classical master equation is satisfied,  nor do we assume that its violation is controlled in some specific way. However, we note that Liu, Sheng \& Xu \cite{Liu:2016} proved that all \emph{regular pre-Courant algebroids} are in fact  \emph{twisted Courant algebroids}.  We will in due course carefully define \emph{symplectic almost Lie 2-algebroids} and show that under a technical assumption on the grading  they are in one-to-one correspondence with pre-Courant algebroids as defined by Vaisman \cite{Vaisman:2005}.\par
At the risk of getting slightly ahead of ourselves, a symplectic almost Lie 2-algebroid is essentially a `symplectic Lie 2-algebroid' for which $\{\Theta, \Theta \} \neq 0$, but the weaker condition $\{ \{\Theta, \Theta\},f\} =0$ holds for all weight zero functions $f$. Here, $\Theta$ is the degree three Hamiltonian  that,  together with the weight ${-}2$ Poisson bracket, encodes  the pre-Courant algebroid. Although we have lost the Jacobi identity for the Courant--Dorfman bracket we still have compatibility of the anchor map with the bracket. This is the origin of our nomenclature `almost'. \par
 An \emph{almost Lie algebroid} (see for example  \cite{Grabowska:2008, Grabowski:2011, Grabowski:1999}) is a `Lie algebroid' for which the Jacobi identity for the bracket on the space of sections has been discarded, but compatibility of the anchor with the bracket remains intact. Such `weakened' Lie algebroids  naturally appear in geometric mechanics with nonholonomic constraints, and  this fact provides motivation  for the study of almost Lie algebroids and related structures. There is an even looser notion of a \emph{skew algebroid} where the compatibility of the anchor map with the bracket is lost. However,  the `almost' case is much better behaved than the `skew' case for some quite fundamental reasons. First, while for both almost Lie algebroids and skew algebroids one has a clear notion of admissible paths,  in order to develop homotopies of admissible paths compatibility of the anchor with the bracket is needed.  Secondly, while one cannot fully develop the cohomology theory of almost Lie algebroids, their zeroth and first cohomology groups are perfectly well defined. For general  skew algebroids, only the zeroth cohomology group can be defined.  The situation  for pre-Courant algebroids or better put symplectic almost Lie 2-algebroids, is very similar to the case of almost Lie algebroids. Moreover, we will show that the relations between Lie algebroids and Courant algebroids pass over to almost Lie algebroids and pre-Courant algebroids.  In particular,  (sub-)Dirac structures in pre-Courant algebroids give natural examples of almost Lie algebroids.  The  general meta-theorem here is that \emph{the Jacobi identity can be thrown out, but in order to have interesting and useful structures compatibility of the anchor with the bracket must be kept!} \par
Another direction of `categorification' is the study of `double objects' in the sense of Ehresmann. For example, double vector bundles are vector bundles in the category of vector bundles. Other examples can then be  built by adding further structures on one or both of the vector bundle structures. For instance, \emph{VB-algebroids} and \emph{VB-Courant algebroids} are Lie algebroids and Courant algebroids in the category of vector bundles respectively (or via `categorical nonsense', vice versa). The notion of a VB-Courant algebroid is less well known than that of the notion of a VB-algebroid  and was introduced by  Li-Bland  in his Ph.D. thesis \cite{Li-Bland:2012} (also see Jotz Lean \cite{JotzLean:2015} and Lang, Sheng \& Wade \cite{Lang:2018}). We  will define VB-pre-Courant algebroids, as well as their higher graded  versions\footnote{See \cite{Bruce:2015} for the notion of weighted Lie algebroids and weighted Lie groupoids.} in simple terms of pre-Courant algebroids equipped with a compatible homogeneity structure \cite{Grabowski:2009, Grabowski:2012}.  Loosely, a  \emph{weighted pre-Courant algebroid} is a pre-Courant algebroid with a compatible $\mathbb{N}$-grading. In fact, in  \cite{Bruce:2015} we (together with Grabowska) made a very brief incursion into the theory of  weighted Courant algebroids, but here we will explain the relation with VB-Courant algebroids  explicitly.  Using homogeneity structures, i.e., smooth actions of the multiplicative monoid of reals, we simultaneously simplify and generalise  the notion of a VB-Courant algebroid. For example, in this framework  the notion of a weighted, or just a VB,  pre-Dirac structure is conceptually clear. \par
We remark that much of our inspiration for working with double and multiple structures comes from Mackenzie's works on double vector bundles, VB-algebroids etc., in  relation to Poisson geometry, see for example  \cite{Mackenzie:1992,Mackenzie:2000}.

\smallskip

\noindent \textbf{Main Results}: \par
\begin{itemize}
\item We show that there is a one-to-one correspondence between pre-Courant algebroids (for which the Grassmann parity and weight coincide)  and symplectic almost Lie 2-algebroids  (Theorem  \ref{thm:preCourant}).
\item We prove that sub-Dirac structures on pre-Courant algebroids are almost Lie algebroids  (Proposition \ref{prop:Dirac}).
\item We show that there exists  a quasi-cochain map between the na\"{\i}ve quasi-cochain complex and the standard quasi-cochain complex  (Theorem \ref{thm:quasimap}).
\item We prove that the higher order tangent bundles of a pre-Courant algebroid are canonically weighted pre-Courant algebroids  (Theorem \ref{thm:highertangentlift}).
\end{itemize}
We also present several illustrative examples of the various structures encountered in this paper.

\smallskip

\noindent\textbf{Our use of supermanifolds}: We will assume that the reader has some basic working knowledge of the theory of supermanifolds. Formally we understand a supermanifold to be a  locally  ringed space $M: = \big( |M| , \cO_M\big)$ where  the topological space $|M|$ is  Hausdorff  and second countable, and $\cO_M$ is a sheaf of commutative superalgebras  such that  on `small enough' open subsets $|U| \subset |M|$,  is locally isomorphic to $C^\infty(\R^p) \otimes \wedge^\bullet \R^q$.  We will regularly employ local coordinates throughout this paper.  We will denote the Grassmann parity of various objects using `tilde'. The parity reversion functor we will denote as $\Pi$.\par

Although we generally work with supermanifolds, our  intention is not to generalise pre-Courant algebroids to `super pre-Courant algebroids', but rather to show how the framework of symplectic supermanifolds offers new light on the subject of classical pre-Courant algebroids.

\smallskip

\noindent \textbf{Arrangement}: This paper is arranged as follows. In Section \ref{sec:Preliminaries} we recap the basic notions needed in this paper.  In Section \ref{sec:Almost} we present our main definitions of symplectic almost Lie 2-algebroids and how they are related to pre-Courant algebroids. In Section \ref{sec:NGrading} we provide a definition of a pre-Courant algebroid carrying an additional compatible graded structure and show how these generalise VB-Courant algebroids. We finish this paper in Section \ref{sec:Con} with a few concluding remarks.

\section{Preliminaries}\label{sec:Preliminaries}
In this section, we sketch some background material needed throughout the rest of this paper.  The informed reader may safely skip this section.
\subsection{Pre-Courant algebroids}\label{sec:PreCourant}
 Let $(E, \langle \cdot, \cdot \rangle)$ be a pseudo-Euclidean vector bundle over a smooth manifold $M$. The metric induces an isomorphism between $E$ and its dual $E^{\ast}$ in the standard way.
\begin{definition}[Vaisman \cite{Vaisman:2005}]
A \emph{Courant vector bundle} is a pseudo-Euclidean vector bundle equipped with an anchor map
$$\rho: E \rightarrow \sT M,$$
such that $\ker(\rho)$ is a coisotropic distribution in $E$.
\end{definition}
Given any Courant vector bundle $(E, \langle \cdot, \cdot \rangle, \rho)$,  we have a differential operator $\mathcal{D} : C^{\infty}(M) \rightarrow \Sec(E)$ defined  by
\begin{equation}\label{eqn:D}
 \langle \mathcal{D}f, e \rangle:=\rho(e)f.
\end{equation}
\begin{definition}[Vaisman \cite{Vaisman:2005}]\label{def:precal}
A \emph{pre-Courant algebroid} is a Courant vector bundle $(E, \langle \cdot, \cdot \rangle, \rho)$ equipped with a $\R$-bilinear product  $ \circ: \Sec(E) \times \Sec(E)\rightarrow \Sec(E)$, which we will refer to as a  \emph{pre-bracket},  that satisfies:
\begin{enumerate}
\item $\rho(e_{1} \circ e_{2})  = [\rho(e_{1}), \rho(e_{2})]$;\label{def:preC1}
\item $ e_{1} \circ e_{2} + e_{2} \circ e_{1} =  \mathcal{D}\langle e_{1}, e_{2}\rangle$;\label{def:preC2}
\item $\rho(e_{1})\langle e_{2}, e_{3}\rangle = \langle e_{1}\circ e_{2}, e_{3}\rangle + \langle e_{2}, e_{1}\circ e_{3}\rangle $,\label{def:preC3}
\end{enumerate}
for all $e_{1}, e_{2}$ and $e_{3} \in \Sec(E)$.
\end{definition}
 Note that in general, a pre-bracket does not satisfy any form of Jacobi identity.   However, if a pre-bracket satisfies the Jacobi identity in Loday--Leibniz form, then we are dealing with a \emph{Courant algebroid} (see \cite{Courant:1990,Dorfman:1987,Liu:1997}).
\begin{example}
Courant algebroids are examples of pre-Courant algebroids by simply disregarding the Jacobi identity for the Courant--Dorfman bracket.
\end{example}
\begin{example}\label{exp:7crossproduct}
It is well known that  Courant algebroids over a point are precisely quadratic Lie algebras. In this setting, the anchor is trivially zero and the defining axioms of a Courant algebroid reduce to the definition of a Lie algebra  equipped with an  invariant quadratic polynomial. If we replace a quadratic Lie algebra with a quadratic skew algebra (i.e., a vector space satisfying the usual axioms of Lie algebras except for the Jacobi identity), then we obtain an algebraic example of a pre-Courant algebroid. The seven-dimensional cross product
\begin{align*}
& \times: \R^{7}\times \R^{7} \rightarrow \R^{7}\\
& (\textbf{a}, \textbf{b}) \mapsto \textbf{a} \times \textbf{b},
\end{align*}
provides an explicit example of such a structure. Recall that the seven-dimensional  cross product can be defined  using a basis as
$e_{i} \times e_{j} = \sum_{k} \epsilon_{ijk} e_{k},$
where $\epsilon_{ijk}$ is a totally antisymmetric tensor with value $+1$ for $ijk = 123,\; 145, \; 176, \; 246,\: 257, \; 365$. All the remaining components that do not come from a permutation of these values are zero. From this definition (among other identities) we have that:
\begin{eqnarray}
\textbf{a} \times \textbf{b} &=& {-} \textbf{b} \times \textbf{a}; \label{eqn:7Xa}\\
\textbf{a} \cdot (\textbf{b} \times \textbf{c}) &=& \textbf{b} \cdot (\textbf{c} \times \textbf{a}) = \textbf{c} \cdot (\textbf{a} \times \textbf{b}),\label{eqn:7Xb}
\end{eqnarray}
where the `dot' is the  standard Euclidean scalar product. However, in contrast to the three-dimensional cross product, the Jacobi identity is not satisfied. Comparing \eqref{eqn:7Xa} and  \eqref{eqn:7Xb} with the axioms of Definition \ref{def:precal}, we see that we do indeed obtain a pre-Courant algebroid over a point.
\end{example}
 The subspaces  of $\Sec(\wedge^{\bullet} E^{\ast})$ defined as
$$C^{k}(E) = \left\{\Psi \in \Sec(\wedge^{k} E^{\ast})~ | ~~ i_{\mathcal{D}f}\Psi = 0 ~~\textnormal{for all} ~ f \in C^{\infty}(M)  \right\},$$
posses a covariant derivative $\mathcal{D}: C^{k}(E) \rightarrow C^{k+1}(E)$ defined as
$$\mathcal{D}\Psi(e_{1}, \cdots e_{k+1}) = \sum_{i=1}^{k+1} (-1)^{i +1} \rho(e_{i})\Psi(e_{1}, \cdots , \hat{e}_{i} , \cdots e_{k+1}) + \sum_{i<j} (-1)^{i+j}\Psi(e_{i}\circ e_{j} , e_{1}, \cdots, \hat{e}_{i}, \cdots , \hat{e}_{j}, \cdots ,e_{k+1}),$$
which is the natural extension on the operator \eqref{eqn:D}. Following Sti\'{e}non \& Xu \cite{Stienon:2008} and Liu, Sheng \& Xu \cite{Liu:2016} one can construct a quasi-cochain complex associated with a pre-Courant algebroid.
\begin{definition}[Liu, Sheng \& Xu \cite{Liu:2016}]\label{def:naivequasicomplex}
The \emph{na\"{\i}ve quasi-cochain complex } of a pre-Courant algebroid is the quasi-cochain complex
$$(C^{\bullet}(E), ~ \mathcal{D}).$$
\end{definition}
It is important to note that in general  $\mathcal{D}$ is not `homological', i.e., $\mathcal{D}^{2} \neq 0$,  due to the fact that the pre-bracket does not satisfy the Jacobi identity. Thus, we have only a quasi-cochain complex and not a genuine cochain complex.

\subsection{Graded bundles}
Our general understanding of a `graded supermanifold' is in the spirit of Th.~Voronov \cite{Voronov:2001}. \emph{Graded manifolds} are defined as supermanifolds  equipped with a privileged class of atlases in which the  coordinates are assigned weights in $\mathbb{Z}$, which are in general independent of the Grassmann parity. Moreover, the coordinate changes are decreed to be polynomial in non-zero weight coordinates and must respect the weight. An additional condition is that all the non-zero weight coordinates that are Grassmann even must be `cylindrical'.  In precise terms, it means that the associated \emph{weight vector field is complete} \cite[Definition 2.2]{Grabowski:2013}.\par
An important class of graded  manifolds are those that carry non-negative grading.  For the moment  we will only consider  classical smooth manifolds, although the statements here will generalise to the category of supermanifolds (see \cite{Jozwikowski:2016}). We will  require that the grading on the structure sheaf of a (smooth) manifold $F$ is associated with a smooth action $\rmh:\R\ti F\to F$ of the monoid $(\R,\cdot)$ of multiplicative reals. Such an action is referred to as a \emph{homogeneity structure} in the terminology of Grabowski \& Rotkiewicz   \cite{Grabowski:2012}. This action, when reduced to $\R_{>0}$, is the one-parameter group of diffeomorphism integrating the \emph{weight vector field}, thus the weight vector field is  \emph{h-complete}, see Grabowski \cite[Definition 2.2]{Grabowski:2013}. As a  consequence of this, all  homogeneous functions must be of \emph{non-negative integer weight}. Thus, the algebra $\mathcal{A}(F)\subset C^\infty(F)$ spanned by homogeneous functions is $\mathcal{A}(F) =  \bigoplus_{i \in \mathbb{N}}\mathcal{A}^{i}(F)$, where $\mathcal{A}^{i}(F)$ consists of homogeneous functions of degree $i$.\par
Importantly,  for $t \neq 0$ the action $\rmh_{t}$ is a diffeomorphism of $F$, and when $t=0$ it is a smooth surjection $\tau=\rmh_0$ onto $F_{0}=: M$, with the fibres being diffeomorphic to $\mathbb{R}^{N}$ (or in the supercase $\mathbb{R}^{N|P}$).  Thus, the objects obtained are particular kinds of \emph{polynomial bundles} $\tau:F\to M$,  i.e., fibrations which locally look like $U\times\R^N$  ($U \subseteq M$ open) and the changes of coordinates (for a certain choice of an atlas) are polynomials on $\R^N$. For this reason, graded manifolds with non-negative weights \emph{and} h-complete weight vector fields  $\zD \in \Vect(F)$ are also known as \emph{graded bundles} (see \cite{Grabowski:2012}).  \par
On a  graded bundle $F$, one can always choose an  atlas   consisting of charts for which we  have homogeneous local coordinates $\big(x^{a}, y^{I}\big)$, where $\w(x^{a}) =0$. It is convenient to add a redundant index $w$ to indicate the weight of $y^I$. So our coordinates will be $\big(x^{a}, y_w^{I}\big)$, where $\w(y_{w}^{I}) = w$  with $1\leq w\leq k$, for some $k \in \mathbb{N}$ known as the \emph{degree}  of the graded bundle.  Note that, according to this definition, a graded bundle of degree $k$ is automatically a graded bundle of degree $l$ for $l\ge k$. However, there is always a \emph{minimal degree}.  A  little more explicitly, the changes of local coordinates are of the form
\begin{align}\label{eqn:translawsGr}
& x^{a'} = x^{a'}(x), & y^{I'}_{w} = \sum \frac{1}{n!} y^{J_{1}}_{w_{1}} y^{J_{2}}_{w_{2}} \cdots y^{J_{n}}_{w_{n}}T_{J_{n} \cdots J_{2} J_{1}}^{\:\:\:\:\:\:\:\:\:\:\:\: \:\:\:\: I'}(x),
\end{align}
\noindent where $w=w_1+\ldots+w_n$ and we assume the tensors $T_{J_{n} \cdots J_{2} J_{1}}^{\:\:\:\:\:\:\:\:\:\:\:\: \:\:\:\: I'}$ to be graded-symmetric in lower indices.  Naturally, changes of coordinates are invertible and so, automatically, $\left(T_{J}^{\:\: I'}(x)\right)$ is an invertible matrix. A graded bundle  of degree $k$  admits a sequence of  surjections
\begin{equation*} 
F=F_{k} \stackrel{\tau^{k}_{k-1}}{\longrightarrow} F_{k-1} \stackrel{\tau^{k-1}_{k-2}}{\longrightarrow}   \cdots \stackrel{\tau^{3}_2}{\longrightarrow} F_{2} \stackrel{\tau^{2}_1}{\longrightarrow}F_{1} \stackrel{\tau^{1}}{\longrightarrow} F_{0} = M,
\end{equation*}
\noindent where $F_l$ itself is a graded bundle over $M$ of degree $l$ obtained from the atlas of $F_k$ by removing all coordinates of degree greater than $l$. \par
A homogeneity structure is said to be \emph{regular} if and only if
\begin{equation*}
\left.\frac{\rmd }{\rmd t}\right|_{t=0}\rmh_{t}(p) = 0  \hspace{20pt} \Longrightarrow \hspace{20pt} p = \rmh_{0}(p)  \in M,
\end{equation*}
 for all points $p \in F$ (see \cite{Grabowski:2009}). Moreover, if the homogeneity structure is regular, then the graded bundle is, in fact, a vector bundle. The converse is also true, and  so  naturally $\tau^{1}:F_{1} \rightarrow M $ is a vector bundle.\par
\begin{example}\label{exm:highertangent}
The fundamental  example of a graded bundle is the higher tangent bundle $\sT^{k}M$, i.e., the bundle of $k$-th jets (at zero) of curves $\gamma: \mathbb{R} \rightarrow M$.
Given a smooth function $f$ on a manifold $M$, one can construct functions $f^{(\alpha)}$ on $\sT^k M$, where $0\leq \alpha\leq k$, which are referred to as the $(\alpha)$-lifts of $f$ (see \cite{Morimoto:1970}). These functions are  defined by
$$
f^{(\alpha)}([\gamma]_k):= \left.\frac{\rmd^\alpha}{\rmd t^\alpha}\right|_{t=0} f(\gamma(t)),
$$
where $[\gamma]_k\in \sT^k M$ is the class of the curve $\gamma:\R\rightarrow M$.
The smooth  functions $f^{(k)} \in C^\infty\big(\sT^k M \big)$ and $f^{(1)} \in C^\infty\big(\sT M \big)$ are called the $k$-complete lift and the tangent lift of $f$, respectively. A coordinate system $(x^a)$ on $M$ gives rise to so-called \emph{adapted} or \emph{homogeneous} coordinate system  $(x^{a, (\alpha)})_{0\leq \alpha \leq k}$ on $\sT^k M$ in which $x^{a, (\alpha)}$ are of weight $\za$. Fa\`{a} di Bruno's formula, i.e., repeated application of the chain rule, shows that the admissible changes of adapted coordinates are of the form \eqref{eqn:translawsGr}.
\end{example}

Moving on to the supermanifold case, if the Grassmann parity of the homogeneous local coordinates is given by the weight mod 2, then we speak of an \emph{N-manifold} (cf. \cite{Roytenberg:2002, Severa:2005}). In this case $\rmh_{-1}$ acts as the parity operator, i.e., it flips the sign of any Grassmann odd function. \par
 In general, we will not assume that the Grassmann parity has any direct relation to the assignment of weight. We  only assume that the weight(s) and the Grassmann parity are \emph{compatible} in the sense that the action of the homogeneity structure and the parity operator commute, i.e, we can find local coordinates that have both a well-defined weight and Grassmann parity. \par
Morphisms between graded bundles are smooth maps which preserve the weight, so we obtain the category of graded bundles.   In other words, morphisms relate the respective homogeneity structures, or equivalently morphisms relate the respective weight vector fields. \par
Following \cite{Grabowski:2009, Grabowski:2012}, the notion of  \emph{double} and $n$-\emph{fold graded bundles} is clear: we have a (super)manifold equipped with $n\ge 2$ homogeneity structures that commute. For example, if we have a double graded bundle, then we have two homogeneity structures, $\rmh^{1}$ and $\rmh^{2}$ (say), such that
$$\rmh^{1}_{s}\circ \rmh^{2}_{t} =  \rmh^{2}_{t}\circ \rmh^{1}_{s},$$
for all $s$ and $t \in \mathbb{R}$.
If both the homogeneity structures are regular, then we have a double vector bundle. The obvious statements hold for any number of regular homogeneity structures. For further details, the reader is encouraged to consult \cite{Grabowski:2013,Grabowski:2012,Roytenberg:2002, Severa:2005,Voronov:2001}.

\subsection{Pseudo-Euclidean vector bundles} \label{sec:PseudoE}

 It is well known that any pseudo-Euclidean vector bundle $\big(E, \: \langle ,\rangle \big)$ (in the category of smooth manifolds) has a \emph{minimal symplectic realisation} as a symplectic N-manifold of degree 2 (see \cite{Roytenberg:2002}).  The supermanifold $\sT^{*} \Pi E$ -- which is canonically fibred over $\Pi(E\oplus E^{\ast}) $ --  is a bi-graded bundle, where we use  the phase-lift (see \cite{Grabowski:2013}) of the homogeneity structure on $\Pi E$. By passing to the total weight we obtain an N-manifold of degree 2 which carries a canonical symplectic structure. The minimal symplectic realisation of the pseudo-Euclidean vector bundle, which we denote as $F_{2}$, is given by the pullback of $\sT^{*} \Pi E$ with respect to the embedding $E \hookrightarrow E \oplus E^{*}$ given by $X \mapsto \big(X, \langle \half X, \bullet \rangle \big)$. In particular, the symplectic N-manifold completes the following diagram.
\begin{center}
\leavevmode
\begin{xy}
(0,15)*+{F_{2}}="a"; (25,15)*+{\sT^{\ast}\Pi E}="b";%
(0,0)*+{\Pi E}="c"; (25,0)*+{\Pi (E \oplus E^{\ast})}="d";%
{\ar "a";"b"}?*!/^3mm/{};%
{\ar "a";"c"}?*!/^3mm/{};{\ar "b";"d"}?*!/_3mm/{};%
{\ar "c";"d"} ?*!/^3mm/{};%
\end{xy}
\end{center}

\smallskip

 Furthermore, it is well  known that one can always find affine Darboux charts consisting of local coordinates  $(x^{a}, \zx^{i}, p_{b})$  of weight $0, \: 1$ and $2$, respectively, such that the symplectic form is given by
\begin{equation}\label{eqn:SymplecticForm}
\omega =  \rmd p_{a} \rmd x^{a} + \frac{1}{2}\rmd \zx^{i}g_{ij}\rmd\zx^{j},
\end{equation}
 where $g_{ij}$ corresponds to the pseudo-Euclidean structure on $E$. In particular, the local coordinates $\zx^{i}$ correspond to local basis vectors  $e_{i} \in \Sec(E)$, chosen  such that  $\langle e_{i} , e_{j} \rangle = g_{ij} = constant$.  Note that the Grassmann parity of the coordinates correspond to the weight mod $2$. Thus, $\zx^i$ are Grassmann odd, while $x^a$ and $p_b$ are Grassmann even. \par
The changes of affine Darboux coordinates induced by change of local trivialisations of $E$,
\begin{align*}
& x^{a'} = x^{a'}(x), && e_{i'} = T_{i'}^{\:\:j}(x)e_j,
\end{align*}
such that $T_{i'}^{\:\: j}g_{jk}T_{l'}^{\:\:k} = g_{i'l'} = constant$,
are of the form
\begin{align}\label{eqn:TransDarboux}
& x^{a} = x^{a}(x'),\\
& \zx^i = \zx^{j'} T_{j'}^{\:\:i},  \nonumber \\
& p_a = \left(\frac{\partial x^{b'}}{\partial x^a}\right)  p_{b'} + \frac{1}{2!} \zx^{k'}\zx^{l'} \left(\frac{\partial T_{l'}^{\:\: n}}{\partial x^a} \right)g_{nm}T_{k'}^{\:\: m}.\nonumber
\end{align}
Note that, due to the `quadratic' nature of the transformation law for the `momenta' $p_a$, we have a graded super bundle of degree $2$ for which the weight corresponds to the Grassmann parity of the coordinate.\par
In affine Darboux coordinates the associated non-degenerate Poisson bracket is given by
\begin{equation}\label{eqn:PoissonBracket}
\{X, Y \} = \frac{\partial X}{\partial p_{a}} \frac{\partial Y}{\partial x^{a}} - \frac{\partial X}{\partial x^{a}}\frac{\partial Y}{\partial p_{a}} + (-1)^{\widetilde{X}+1}g^{ij}\frac{\partial X}{\partial \zx^{j}}\frac{\partial Y}{\partial \zx^{i}},
\end{equation}
for $X$ and $Y \in \A^{\bullet}(F_{2})$,  where $(g^{ij})=(g_{ij})^{-1}$ is the inverse of the metric. Via  affine Darboux coordinates, it is clear that the Poisson bracket is of weight ${-}2$.\par
Functions of weight one on $F_{2}$ are identified with sections of $E^{\ast}$, noting the  shift in Grassmann parity. As we have a pseudo-Euclidean structure, we have the canonical identifications
$$\A^{1}(F_{2}) \simeq \Sec( E^{\ast}) \simeq \Sec(E),$$
 of locally free $C^\infty(M)$-modules. More specifically, we have the natural identification $\A^{1}(F_{2}) \simeq \Sec( E^{\ast})$ via the one-to-one mapping  $e^i \mapsto \zx^i$, where  $(e^i)$ is a local basis of $\Sec( E^{\ast})$. Note that this mapping is Grassmann odd, but that this does not affect the module structure at all as we are considering the underlying structure of a vector bundle in the category of manifolds. We can then use the metric to raise and lower indices. That is, we have a diffeomorphism
$$\Pi E^\ast \simeq \Pi E,$$
given by
$$\chi_j = \zx^ig_{ij},$$
where $\chi_j$ serve as fibre coordinates on $\Pi E^\ast$. Thus any section of $E^\ast$ can be considered as a section of $E$  by writing
$$\sigma = \zx^i \sigma_i(x) = g^{ji}\sigma_i(x) \chi_j.$$
The converse identification is clear.   Throughout this paper, we will use this identification of sections of $E$ and $E^\ast$ with weight one functions on $F_2$.

\begin{remark}\label{rem:SuperVector}
The preceding constructions generalise to super vector bundles quite directly via minor modifications. Consider a vector bundle $E \rightarrow M$ in the category of (smooth) supermanifolds. As before one can apply the parity reversion functor.  The  vector bundle $\Pi E \simeq F_1$ can be equipped with local coordinates $(x^a, \zx^i)$   as before, but now we allow these coordinates to be even or odd independently of the degree. In particular, $x^a$ serve as coordinates on $M$ which is now a supermanifold itself. The supermanifold $F_2$  is again  assumed to come equipped  with  an even symplectic form  of weight $2$. Due to a version of Darboux theorem, we can  equip $F_2$ with additional coordinates $p_a$ which  have the same parity as $x^a$, and the symplectic structure is locally of the same form as \eqref{eqn:SymplecticForm}. Note that as $g_{ij}$ is constant,  $g_{ij} \neq 0$ only if $\zx^i$ and $\zx^j$ have the same parity. Moreover, $g_{ij} =g_{ji}$ if the parity is odd and $g_{ij} = {-}g_{ji}$ if this parity is even. \par
Such symplectic forms correspond to even non-degenerate bilinear forms on $\Sec(E)$,  or equivalently on $\Sec(E^\ast)$,  which are (locally) represented by $g_{ij}$.  The  inverse structure $g^{ij}$ is again supersymmetric and even.  Note that the `evenness' in this context  implies that the even and odd subspaces of $\Sec(E)$ are orthogonal. In the other direction, given a super vector bundle equipped with an even non-degenerate bi-linear form one can construct the minimal symplectic realisation in the same way as  the classical case. Thus, the local description is more-or-less identical to the case of symplectic N-manifolds of degree $2$, with the exception that the coordinates may now take on various Grassmann parities and $g_{ij}$ splits into  skew-symmetric and symmetric parts. The only real  effect of  passing to the underlying structure of a super vector bundle will be further sign factors in all the local expressions.
\end{remark}

\subsection{Almost Lie algebroids}
As we will encounter skew and almost Lie algebroids in this paper it is worth recalling the definitions (see for example \cite{Grabowski:2011,Grabowski:1999}).
\begin{definition}\label{def:almostLie}
A \emph{skew algebroid}  is a triple $(E, \rho, [-,-])$ where:
\begin{enumerate}
\item $\pi : E \rightarrow M$ is a vector bundle (in the category of supermanifolds);
\item $\rho : E \rightarrow \sT M$  is a vector bundle morphism covering the identity on $M$, known as the anchor;
\item $[-,-] : \Sec(E) \times \Sec(E) \rightarrow \Sec(E)$ is a  (graded) skew-symmetric bilinear map, i.e.,
$$ [u, v] =  {-}(-1)^{\widetilde{u} \widetilde{v}}[v,u],$$
\end{enumerate}
such that the Leibniz rule
$$[u, f\:v] = \rho(u)(f) \: v  + (-1)^{\widetilde{f} \widetilde{u}} f [u,v], $$
holds for all $u$ and $v \in \Sec(E)$ and $f \in C^{\infty}(M)$.  An \emph{almost Lie algebroid} is a skew algebroid such that the further condition:
\begin{enumerate}
\setcounter{enumi}{3}
\item $ \rho([u,v]) = [\rho(u) , \rho(v)]$,
\end{enumerate}
holds for  all $u$ and $v \in \Sec(E)$.  That is, the bracket and anchor are compatible.
\end{definition}
Note that we do \emph{not} insist that the bracket on the space of sections of a skew or almost Lie algebroid is a Lie bracket, i.e.,  in general, the bracket  will not satisfy the Jacobi identity. If the Jacobi identity holds for the bracket on the space of sections of skew algebroid, then we, in fact, have  a \emph{Lie algebroid}. It should be noted that the Jacobi identity for the bracket together with the  Leibniz rule implies that the bracket and anchor are compatible.    \par
Recall that   Lie algebroids can  be described in terms of Q-manifolds following Va\u{\i}ntrob \cite{Vaintrob:1997}. In particular, it is well known that a Lie algebroid structure on a (super) vector bundle $E$ is equivalent to the existence of a weight one homological vector field on the supermanifold $\Pi E$.   Skew algebroids and almost Lie algebroids can also be formulated in similar terms.  Specifically, a skew algebroid structure on a (super) vector bundle $E$ is equivalent to a Grassmann odd vector field of weight one on $\Pi E$. Note that for skew algebroids there is no further condition on the vector field.  For the case of almost Lie algebroids, we need to encode the  compatibility of the bracket and the anchor.  The reader can easily check the following proposition.
\begin{proposition}\label{prop:almostLieQ}
There is a natural one-to-one correspondence between  the two following classes of objects:
 \begin{enumerate}
 \item almost Lie algebroid structures on the (super) vector bundle $\pi: E \rightarrow M$;
 \item  odd vector fields $\rmd_{E}$ of weight one on the supermanifold $\Pi E$ that  satisfy
$$d_{E}^{2}f =0,$$
for all $f \in C^{\infty}(M)$.
 \end{enumerate}
\end{proposition}
\noindent The bracket and anchor can be recovered using the derived bracket formalism of Kosmann-Schwarzbach \cite{Kosmann-Schwarzbach:2003}. We will also encounter almost Lie algebroids that carry a compatible non-negative grading (see \cite{Bruce:2015,Bruce:2016}). For the purposes of this paper it will be convenient to use the following definition.
\begin{definition}\label{def:weightedalmostLie}
A \emph{weighted almost Lie algebroid} of degree $k$ is a double graded superbundle   $E_{(k-1,1)}$ such that the super vector bundle $\pi : E_{(k-1,1)} \rightarrow M_{(k-1,0)} $ is an almost Lie algebroid  with the property that the bracket is of weight $-k$, i.e., for two homogeneous sections of weight $r_{1}$ and $r_{2}$ the resulting bracket is a homogeneous section of weight $r_{1} + r_{2} - k$.
\end{definition}
Let us employ local coordinates $\big(x^{\alpha}_{w}, Y^{\Sigma}_{w} \big)$ on $E_{(k-1,1)}$,  here we have split the coordinates into base and fibre coordinates with respect to the vector bundle structure. Let us denote some chosen basis   of the (local) sections by $(e_{\Sigma}^{w})$. Then,  a homogeneous section of weight $r$ is locally of the form
$$u[r] = u^{\Sigma}[r-1 -w](x,y)e_{\Sigma}^{w},$$
as we consider a basis element to carry weight $(w,1)$.  Here the notation $u^{\Sigma}[r-1 -w](x,y)$ means the homogeneous part of the component of weight $r-1-w$. If this component is seemingly of negative weight we understand it to be zero.

\section{Symplectic almost Lie 2-algebroids and pre-Courant algebroids}\label{sec:Almost}
\subsection{Symplectic nearly and almost Lie 2-algebroids}\label{subsec:SymplecticAlmost}
Taking our cue from Roytenberg \cite{Roytenberg:2002}, we propose the  following (preliminary) definition.
\begin{definition}\label{def:nearlyLie2}
A \emph{symplectic nearly Lie 2-algebroid} consist of the following data:
\begin{enumerate}
\item A degree $2$ graded superbundle $(F_{2},~ \rmh)$;
\item A nondegenerate Poisson bracket $\{ -,-\}$ of weight ${-}2$;
\item An odd function $\Theta \in \A^{3}(F_{2})$.
\end{enumerate}
\end{definition}
Note that we do \emph{not} require any further condition on the function $\Theta$, in particular, we do not insist  that the classical master equation holds, i.e., in general  $\{\Theta, \Theta \} \neq 0$. This means that the Hamiltonian vector field
$$Q_{\Theta} = \{\Theta, \:-\:  \} \in \Vect(F_{2}),$$
\noindent is \emph{not} in general homological. Thus, we will refer to $\Theta$ as a \emph{pre-homological potential}.  Moreover, we do not insist that $F_{2}$ be an N-manifold. Thus, the assignment of Grassmann parity and weight to  geometric objects on $F_2$, e.g., functions and vector fields,  may be quite independent.
\begin{remark}
We understand a \emph{Lie $n$-algebroid} to be a graded superbundle of degree $n$ equipped with a Grassmann odd homological vector field of weight $1$. As we are not restricting ourselves to N-manifolds these supermanifolds also appear under the name \emph{non-linear Lie algebroids} in the works of Voronov \cite{Voronov:2001,Voronov:2012}.  Understanding Lie $n$-algebroids in terms of brackets and anchors is harder work than  the classical case of Lie algebroids, see Bonavolont\`{a} \& Poncin \cite{Bonavolonta:2013} for details. A \emph{symplectic Lie $n$-algebroid} is then a Lie $n$-algebroid equipped with a symplectic structure such that the homological vector field is  a Hamiltonian vector field, see \v{S}evera \cite{Severa:2005,Severa:2017}.
\end{remark}
\begin{remark}
 Recall that multivector fields on a manifold $M$ can be identified with functions on the supermanifold $\Pi \sT^{*}M$, which is known as the \emph{anticotangent bundle}. Note that all functions on the anticotangent bundle of a manifold are automatically polynomial in the fibre coordinates.  In particular, vector fields are identified with linear functions on $\Pi \sT^{*}M$ which are necessarily Grassmann odd. The  odd symplectic bracket on $\Pi \sT^{*}M$ is identified, up to a shift in Grassmann parity, with the Schouten--Nijenhuis bracket  on multivector fields. Restriction of the odd symplectic bracket to linear functions is identified with the standard Lie bracket on vector fields. One should keep the description of the  Lie bracket on the space of vector fields on a smooth manifold in terms of an odd symplectic bracket in mind when reading the rest of this paper, it will help  to explain some of our notation and the appearance of certain minus signs.
\end{remark}
 Sections of $E$ are naturally identified with weight one functions on $F_2$, i.e., $\A^1(F_2) \simeq \Sec(E)$ as locally free $C^\infty(M)$-modules. We will refer to elements of $\A^1(F_2)$ as sections taking into account the already stated identification.  From  the data of a symplectic nearly Lie 2-algebroid we construct  the following structures on the module of sections of $E$,  using the derived bracket formalism (see  Kosmann-Schwarzbach \cite{Kosmann-Schwarzbach:1996, Kosmann-Schwarzbach:2003} and Voronov  \cite{Voronov:2005}). For any $\sigma$ and $\psi \in \A^1(F_2)$, and $f \in C^\infty(M)$ we have:
\begin{enumerate}
\item A \emph{Courant--Dorfman pre-bracket}
$$\SN{\sigma, \psi}_{\Theta} := (-1)^{\widetilde{\sigma}}\{\{\Theta, \sigma \}, \psi \} ;$$
\item An \emph{anchor map} $\rho_{\Theta}(-) :\Sec(E) \rightarrow \Vect(M)$
$$\rho_{\Theta}(\sigma)f := (-1)^{\widetilde{\sigma}} \{ \{\Theta, \sigma \},f \};$$
\item A \emph{nondegenerate pairing} $\langle - , -  \rangle  : \Sec(E) \times \Sec(E) \rightarrow C^\infty(M)$
$$ \langle \sigma, \psi \rangle :=  \{\sigma, \psi \} .$$
\end{enumerate}
Note that the Courant--Dorfman pre-bracket is \emph{not} a  Lie bracket. However, it is a derived bracket and as such  we have the following well-known results.  First, we do not have the  standard symmetry property, but rather
\begin{equation}\label{eqn:symmetry}
\SN{\sigma , \psi}_\Theta = {-} (-1)^{(\widetilde{\sigma} +1)(\widetilde{\psi}+1)} \SN{\psi, \sigma}_\Theta + (-1)^{\widetilde{\sigma}} \{\Theta, \langle \sigma, \psi \rangle \}.
\end{equation}
Note that we  have to take into account the shift in Grassmann parity with the identification of sections of $E$ and weight one functions on $F_2$. The  Courant--Dorfman pre-bracket  is closer to a Loday--Leibniz bracket than a Lie bracket with respect to its symmetry (cf. Loday \cite{Loday:1993}). However, in general, we also lose the Jacobi identity.   The Jacobiator we define as the map
$$ J_{\Theta} : \A^1(F_2) \times \A^1(F_2) \times \A^1(F_2)  \longrightarrow \A^1(F_2) $$
given by
\begin{align}\label{eqn:Jacobi}
\nonumber J_{\Theta}(\sigma, \psi, \lambda) &:= \SN{\sigma, \SN{\psi, \lambda}_{\Theta}}_{\Theta} - \SN{ \SN{\sigma, \psi}_{\Theta}, \lambda}_{\Theta}  {-} (-1)^{(\widetilde{\sigma}+1)(\widetilde{\psi}+1)}\SN{\psi, \SN{\sigma, \lambda}_{\Theta}}_{\Theta}\\
& = (-1)^{\widetilde{\psi}}\frac{1}{2!} \{\{ \{ \{ \Theta, \Theta \}, \sigma \} , \psi\}, \lambda \},
\end{align}
 for all  $\sigma ,\psi$ and $\lambda \in \A^1(F_2)$.  The Jacobiator is in general non-vanishing, as   $\{ \Theta, \Theta\} \neq 0$.
\begin{proposition}\label{prop:Anchor}
For any symplectic nearly Lie $2$-algebroid the following  identities hold:
\begin{enumerate}
\item$ \SN{\sigma, f \:\psi}_{\Theta} = \rho_{\Theta}(\sigma)f \: \psi  + (-1)^{\widetilde{f} (\widetilde{\sigma} + 1)} f \: \SN{\sigma , \psi}_{\Theta}\, ,$
\item $ \rho_{\Theta}(\sigma) \langle \psi, \lambda  \rangle = \langle \SN{\sigma, \psi}_{\Theta}, \lambda \rangle + (-1)^{\widetilde{\psi} (\widetilde{\sigma} +1)} \langle \psi, \SN{\sigma, \lambda}_{\Theta} \rangle\, ,$
\end{enumerate}
for all  $\sigma ,\psi$ and $\lambda \in \A^1(F_2)$ and $f \in C^\infty(M)$.
\end{proposition}
\begin{proof}
The first equation follows directly from the Leibniz rule for the Poisson bracket. The second equation follows from the Jacobi identity for the Poisson bracket.
\end{proof}
In general, we  lose compatibility of the anchor map with the pre-bracket structure. Via direct computation, following  what is known about standard Lie algebroids, we see that for any symplectic nearly Lie $2$-algebroid
\begin{equation}\label{eqn:anchor}
J_{\Theta}(\sigma, \psi, f \: \lambda) - (-1)^{\widetilde{f}(\widetilde{\psi} + \widetilde{\sigma})} f \:  J_{\Theta}(\sigma, \psi , \lambda) = [\rho_{\Theta}(\sigma), \rho_{\Theta}(\sigma)]f \: \lambda ~ {-} ~\rho_{\Theta}\left(\SN{\sigma, \psi} \right) f \:\lambda.
\end{equation}
This implies that we have compatibility of the anchor map and the  Courant--Dorfman pre-bracket if the Jacobi identity is satisfied, i.e., $\{\Theta, \Theta \} =0$. However, we can derive a more general condition on the pre-homological potential that still allows for the compatibility of the anchor and Courant--Dorfman pre-bracket.
\begin{proposition}\label{prop:ViolationAnchorBracket}
For any symplectic nearly Lie $2$-algebroid the following  identity holds:
\begin{equation}\label{eqn:thetasquared}
[\rho_{\Theta}(\sigma), \rho_{\Theta}(\psi)]f ~ {-} ~\rho_{\Theta}\left(\SN{\sigma, \psi}_{\Theta} \right) f \:  =  \frac{1}{2}(-1)^{\widetilde{f}(\widetilde{\sigma} + \widetilde{\psi}) + \widetilde{\psi}} \{\{\{\{\Theta, \Theta \} ,f\}, \sigma\} , \psi\},
\end{equation}
for all $\sigma$ and $\psi \in \A^1(F_2)$, and $f\in C^\infty(M)$.  In particular, the anchor and the  Courant--Dorfman pre-bracket are compatible if and only if $\{\{\Theta, \Theta\}, f \} =0$ for all $f \in C^{\infty}(M)$.
\end{proposition}
\begin{proof}
Equation \eqref{eqn:thetasquared}  follows from direct evaluation of the left-hand side of (\ref{eqn:anchor}) in terms of the Poisson bracket  with the use of the Jacobi identity. The calculation is not illuminating and so we omit details. It is clear that if $\{\{\Theta, \Theta\}, f \} =0$, then the right-hand side of \eqref{eqn:thetasquared}  is zero.\par
In the other direction, we can use affine Darboux coordinates $(x, \zx, p)$ (see Subsection \ref{sec:PseudoE} and Remark \ref{rem:SuperVector}).  In affine Darboux coordinates, the pre-homological potential is given by
$$\Theta =  \zx^{i}Q_{i}^{a}(x)p_{a} + \frac{1}{3!}\zx^{i} \zx^{j} \zx^{k}Q_{kji}(x).$$
Thus, using \eqref{eqn:PoissonBracket}, we observe that  $\{\{\Theta, \Theta\},f\}$ is of the form
$$A:= F^a(x)p_a+\frac{1}{2}\zx^i \zx^j G_{ji}(x),$$
which is of weight $2$. 
Hence, it is sufficient to consider
$\{\{A, \chi_i  \},\chi_j \} = \pm G_{ij} =0$  and $\{ \{A,x^a\zx^k\},\zx^k \}=\frac{1}{2}F^a\{\zx^k,\zx^k\}=0$, where $\chi_i := \zx^jg_{ji}$. Non-degeneracy of the Poisson bracket -- specifically the fact that we can assume that $\{\zx^i,\zx^j\}=\pm\zd^{ij}$ --  implies that $\{\{\Theta, \Theta \}, f\} =0$ as required.
\end{proof}

The considerations above suggest the following refinement of the notion of a symplectic nearly Lie 2-algebroid in order to capture the compatibility of the anchor map with the Courant--Dorfman pre-bracket.
\begin{definition}\label{def:almostLie2}
A \emph{symplectic almost Lie 2-algebroid}  is a symplectic nearly Lie 2-algebroid (Definition \ref{def:nearlyLie2}) with the further condition that
$$\{ \{ \Theta,\Theta\}, f\} =0,$$
for  all $f \in C^{\infty}(M)$.
\end{definition}
The Jacobiator associated with the Courant--Dorfman pre-bracket on a symplectic almost Lie $2$-algebroid has some  interesting properties which we state in the following two propositions.
\begin{proposition}\label{prop:SymmJac}
Let $J_\Theta$ be the Jacobiator associated with a symplectic almost Lie 2-algebroid. Then $J_\Theta$ is  totally (graded) skew-symmetric, i.e.,
$$J_{\Theta}(\sigma, \psi, \lambda) = {-}  (-1)^{(\widetilde{\sigma}+1)(\widetilde{\psi} +1)}J_{\Theta}(\psi, \sigma, \lambda) = {-}(-1)^{(\widetilde{\psi}+1)(\widetilde{\lambda} +1)}J_{\Theta}(\sigma, \lambda, \psi),$$
for all $\sigma, \psi$ and $\lambda \in \A^1(F_2)$.
\end{proposition}
\begin{proof}
Directly from the definition of the Jacobiator  and application of the Jacobi identity for the Poisson bracket, we have
 $$J_{\Theta}(\sigma, \psi, \lambda) + (-1)^{(\widetilde{\sigma}+1)(\widetilde{\psi} +1)}J_{\Theta}(\psi, \sigma, \lambda) = (-1)^{\widetilde{\psi}} \{\{\Sigma, \{ \sigma,\psi \}\}, \lambda \}  ,$$
where we use the shorthand $\Sigma = \half \{\Theta, \Theta \}$.  As $\{\psi, \sigma \} = \langle \psi, \sigma \rangle$ is a function on $M$, the right-hand side of the above vanishes (see Definition \ref{def:almostLie2}). Directly from the Jacobi identity we have
$$J_{\Theta}(\sigma, \psi, \lambda) + (-1)^{(\widetilde{\psi}+1)(\widetilde{\lambda} +1)}J_{\Theta}(\sigma, \lambda, \psi) = (-1)^{\widetilde{\psi}} \{ \{ \Sigma, \sigma \}, \{\psi,\lambda \}  \}.$$
 Using the skew-symmetry of the Poisson bracket and the Jacobi identity again, we arrive at
 $$J_{\Theta}(\sigma, \psi, \lambda) + (-1)^{(\widetilde{\psi}+1)(\widetilde{\lambda} +1)}J_{\Theta}(\sigma, \lambda, \psi) =  (-1)^{\widetilde{\psi}}\left(  \{ \Sigma, \{\sigma, \{\psi, \lambda \}  \} \} - \{ \sigma,\{ \Sigma, \{\psi, \lambda \}\}\}\right)=0,$$
as $\{\sigma, \{\psi, \lambda \}  \} $ vanishes due to the weights of $\sigma, \psi$, and $\lambda$, and the other term vanishes as $\{ \psi, \lambda\}$ is a function on $M$, i.e., from Definition \ref{def:almostLie2}, $\{ \Sigma, \{ \psi, \lambda\}\} =0$.
\end{proof}
A section $\sigma$ of a symplectic almost Lie $2$-algebroid with pre-homological potential $\Theta$ is said to be  \emph{$Q_\Theta$-closed} if $Q_\Theta \sigma = 0$. Similarly, a section $\sigma$ is said to be \emph{$Q_\Theta$-exact} if there exists some $f \in C^\infty(M)$ such that $\sigma = Q_\Theta f$.  From Definition \ref{def:almostLie2} it follows that any $Q_\Theta$-exact section is automatically $Q_\Theta$-closed. Using Proposition \ref{prop:ViolationAnchorBracket} and Proposition \ref{prop:SymmJac} we arrive at the following results.
\begin{proposition}\label{prop:JacLin}
Let $J_\Theta$ be the Jacobiator associated with a symplectic almost Lie 2-algebroid. Then
\begin{enumerate}
\item  $J_\Theta$ is $C^{\infty}(M)$-linear in each argument, and
\item  $J_\Theta$ vanishes if at least one of its arguments is $Q_\Theta$-closed.
\end{enumerate}
 \end{proposition}
 \begin{remark}
 Proposition \ref{prop:SymmJac} and part (1) of Proposition \ref{prop:JacLin} were first proved by  Liu, Sheng \&  Xu \cite{Liu:2016} in a traditional framework. We remark that using supermanifolds and the derived bracket construction simplifies the proofs of these statements significantly.
 \end{remark}

\subsection{The relation with pre-Courant algebroids}
The notion of a symplectic almost Lie 2-algebroid  is closely related to the notion of a pre-Courant algebroid as defined by Vaisman  \cite{Vaisman:2005}. Indeed we have the following theorem, which is a minor modification of the theorem of Roytenberg \cite[Theorem 4.5]{Roytenberg:2002}. In order to keep this paper relatively self-contained, we will reproduce the main elements of Roytenberg's arguments.
\begin{theorem}\label{thm:preCourant}
There is a one-to-one correspondence between symplectic almost Lie 2-algebroids for which $F_{2}$ is an N-manifold and pre-Courant algebroids.
\end{theorem}
\begin{proof}
In one direction the correspondence is clear. From a symplectic almost Lie 2-algebroid we construct a pre-Courant algebroid structure on $E \simeq \Pi F_{1} \simeq E^{*}$ (see Subsection \ref{sec:PreCourant} and  Subsection \ref{sec:PseudoE}). In particular,  Proposition \ref{prop:ViolationAnchorBracket} corresponds to Definition \ref{def:precal}  \eqref{def:preC1};   \eqref{eqn:symmetry} corresponds to Definition \ref{def:precal}  \eqref{def:preC2}; and the second equation of Proposition \ref{prop:Anchor} corresponds to Definition \ref{def:precal}  \eqref{def:preC3}. One has to take a little care with the sign factors remembering that the Grassmann parity and weight coincide (mod $2$).\par
In the other direction, there is a one-to-one correspondence between pseudo-Euclidean vector bundles (in the category of smooth manifolds) and symplectic graded bundles of degree $2$ for that have the underlying structure of an N-manifold (see Subsection \ref{sec:PseudoE} and/or \cite{Roytenberg:2002,Severa:2017}).   The non-degeneracy of the Poisson bracket means that given a Courant--Dorfman pre-bracket, as a derived (pre-)bracket, it uniquely defines the pre-homological potential $\Theta$. More specifically, in affine Darboux coordinates $(x^{a}, p_{b}, \zx^{i})$ any weight $3$ pre-homological potential must be of the form
$$\Theta =  \zx^{i}Q_{i}^{a}(x)p_{a} + \frac{1}{3!}\zx^{i} \zx^{j} \zx^{k}Q_{kji}(x).$$
Then, using the Poisson bracket \eqref{eqn:PoissonBracket}, one can directly deduce the relations between the  pre-bracket and pairing,  and the pre-homological potential:
\begin{align*}
& \langle \SN{\chi_{i}, \chi_{j}}_\Theta, \chi_{k} \rangle = Q_{ijk}, & \rho_{\Theta}(\chi_{i}) =  Q^{a}_{i}\frac{\partial}{\partial x^{a}},
\end{align*}
where we set $\chi_i := \zx^{j}g_{ji}$ and remembering that $g_{ij} = constant$ in this privileged class of coordinates.  These local formulas are of course no different from the standard ones for Courant algebroids following Roytenberg \cite{Roytenberg:2002}. Thus, given a pre-Courant algebroid (in the category of manifolds), we can \emph{locally} construct $\Theta$. More than this, one can directly check that this construction is, in fact, global by examining the permissible coordinate changes (\ref{eqn:TransDarboux}).\par
One can also deduce the \emph{structure equations}  for a symplectic almost Lie 2-algebroid (assuming $F_{2}$ is an N-manifold)
\begin{align*} & Q_{i}^{b}\frac{\partial Q_{j}^{a}}{\partial x^{b}} \: {-} \: Q_{j}^{b}\frac{\partial Q_{i}^{a}}{\partial x^{b}} - g^{lm}Q_{mij}Q_{l}^{a} = 0, & g^{ij}Q_{j}^{a}Q_{i}^{b} =0.
\end{align*}
These equations hold if and only if we have a pre-Courant algebroid structure (see Proposition \ref{prop:ViolationAnchorBracket} and (\ref{eqn:thetasquaredlocal})).
\end{proof}

\noindent \textbf{Nomenclature.} Due to Theorem \ref{thm:preCourant}, we will  use the terms \emph{pre-Courant algebroid} and \emph{symplectic almost Lie 2-algebroid} interchangeably in the rest of this paper. In particular, we will use symplectic almost Lie 2-algebroids to \emph{define} pre-Courant algebroid structures on  super vector bundles (see Subsection \ref{sec:PseudoE} and  in particular Remark \ref{rem:SuperVector}).
\subsection{Examples of pre-Courant algebroids}\label{subsec:Examples}
In this subsection,  we present examples of pre-Courant algebroids generally obtained by modifying some of the well-known examples of symplectic Lie $2$-algebroids.
\begin{example}
Continuing Example \ref{exp:7crossproduct}, consider $\sT^{*}\Pi \mathbb{R}^{7}$ which we equip with coordinates $(\zx^{i}, \pi_{j})$ of bi-weight $(1,0)$ and $(0,1)$ respectively. By passing to total weight, it is clear that the canonical Poisson bracket is of total weight $-2$. Moreover, we have $\{\zx^{i}, \pi_{j} \} =  \delta^{i}_{\:\: j}$ and so we recover the standard Euclidean metric on $\mathbb{R}^{7}$. The pre-homological potential in this example is  given by
$$\Theta =  \frac{1}{2}\zx^{i}\zx^{j}\epsilon_{jik} \delta^{kl}\pi_{l}.$$
A quick calculation shows that
$$\SN{\pi_{i} , \pi_{j}}_{\Theta} =  \sum_{k} \epsilon_{ijk}\pi_{k}.$$
Thus, being careful with the shift in Grassmann parity, we recover the seven-dimensional cross product. As the seven-dimensional cross product does not satisfy the Jacobi identity, $\{\Theta, \Theta \} \neq 0$.  As there are no non-zero functions of weight zero on $\sT^{*}\Pi \mathbb{R}^{7}$, it is clear that we have a `symplectic almost Lie 2-algebra'.
\end{example}
\begin{example}\label{exm:twisted}
Let $E\rightarrow M$ be a vector bundle in the category of smooth manifolds. Consider the supermanifold $\sT^{*}\Pi E$, which we equip with homogeneous local coordinates
$$(\underbrace{x^{a}}_{(0,0)},~ \underbrace{\zx^{i}}_{(1,0)},~\underbrace{p_{b}}_{(1,1)}, ~ \underbrace{\pi_{j}}_{(0,1)}).$$
Here  we have used the phase lift of the natural degree-one homogeneity structure on $\Pi E$ to define the graded structure on the cotangent bundle.  By  passing  to total weight we obtain  a degree $2$ symplectic graded bundle.  Let us further assume  that  $E$ is a Lie algebroid, and so $(\Pi E, \rmd_{E})$ is a Q-manifold. The symbol (or Hamiltonian) of the homological vector field $\rmd_E$ is given by
$$\theta  :=  \zx^{i}Q_{i}^{a}(x)p_{a} + \frac{1}{2!}\zx^{i} \zx^{j}Q_{ji}^{k}(x) \pi_{k},$$
and this provides a homological potential on $\sT^{*}\Pi E$.  Thus, we have a  symplectic Lie 2-algebroid structure canonically associated with the underlying Lie algebroid structure on $E$.  Next we  chose some  Lie algebroid $3$-form
$$\alpha =  \frac{1}{3!} \zx^{i} \zx^{j} \zx^{k} Q_{kji}(x),$$
which we do \emph{not} assume to be $\rmd_E$-closed. We can build a  pencil of pre-homological potentials  by defining
$$\Theta_{s} = \theta + s \: \alpha,$$
with $s\in \mathbb{R}$.  A simple calculation shows that for all $f \in C^\infty(M)$
$$\{\Theta_{s}, \{ \Theta_{s}, f\} \} = s \{ \{\theta, \alpha\},f  \} =0,$$
 as $\{\theta, \alpha \}$ is a Lie algebroid form, i.e. there are no conjugate variables in the final expression. Thus, we obtain a family of symplectic almost Lie $2$-algebroids.  If the $3$-form is closed, then we obviously obtain  a family  of  \emph{twisted Courant algebroids}.
\end{example}
\begin{example}\label{exm:bialgebroid}
Similarly to the previous example, suppose that $E$ comes with the structure of an almost Lie algebroid. Furthermore, let us suppose that $E^{*}$ also comes with the structure of an almost Lie algebroid (see Proposition \ref{prop:almostLieQ}). In terms of Hamiltonian functions on $\sT^{*}\Pi E$ we have
\begin{align*}
& \theta :=  \zx^{i}Q_{i}^{a}(x)p_{a} + \frac{1}{2}\zx^{i} \zx^{j}Q_{ji}^{k}(x)\pi_{k}, && \nu := Q^{ai}(x)\pi_{i}p_{a} + \frac{1}{2}\zx^{k}Q_{k}^{ij}(x)\pi_{j}\pi_{i},
\end{align*}
which are of bi-weights $(2,1)$ and $(1,2)$, respectively,  thus both are of total weight $3$. Then,
$$\Theta_{s} = \theta + s \: \nu,$$
with $s \in \mathbb{R}$ provides   $\sT^{*}\Pi E$ with the structure of a pencil of pre-homological potentials.  Now suppose   $\{ \theta, \nu\} =0$, that is we have an \emph{almost  Lie bi-algebroid} (cf. \cite{Kosmann-Schwarzbach:1995,Roytenberg:2002,Voronov:2001}), then a direct calculation shows that
$$\{ \{ \Theta_{s}, \Theta_{s}\} ,f\} =0,$$
for all $f \in C^{\infty}(M)$. Thus, given an almost Lie bi-algebroid we can construct a family of symplectic almost Lie $2$-algebroids.
\end{example}
\begin{example}
Consider a skew algebra  $(\mathfrak{g}, [\bullet, \bullet])$ (i.e., a vector space  equipped with a bracket satisfying the usual axioms of Lie algebras except for the Jacobi identity) and an action thereof on a manifold $M$. We describe the action in the following way. First, we consider the N-manifold
$$F_{2} :=  \sT^{*}M \times \Pi \mathfrak{g} \times \Pi \mathfrak{g}^{*},$$
which we equip with local coordinates
$$(\underbrace{x^{a}}_{0},~ \underbrace{p_{b}}_{2}, ~\underbrace{\zx^{i}}_{1}, \underbrace{\pi_{j}}_{1}  ).$$
The degree ${-}2$ Poisson bracket is just the canonical Poisson bracket on $F_{2} \simeq \sT^{*}(M \times \Pi \mathfrak{g})$. The pre-homological potential is given by
$$\Theta =  {-} \zx^{i}Q_{i}^{a}(x)p_{a} + \frac{1}{2!}\zx^{i}\zx^{j}Q_{ji}^{k}\pi_{k},$$
where $Q_{i}^{a}$  describe the action and $Q_{ji}^{k}$ are the structure constants of the skew algebra. The minus sign is to ensure we have a right action. We can view $\psi =  \psi^{i}\pi_{i}$ as an element of $\mathfrak{g}$ (once we shift the parity) and write the action as
$$\rho_{\Theta}(\psi) = \psi^{i}Q_{i}^{a}(x)\frac{\partial}{\partial x^{a}}.$$
A simple calculation shows  that $\{\Theta , \{ \Theta, f \} \} =0$ for all $f \in C^{\infty}(M)$ and so we have  a symplectic almost Lie $2$-algebroid. The compatibility of the anchor and the pre-bracket means
$$[\rho_{\Theta}(\psi_{1}), \rho_{\Theta}(\psi_{2})] =  \rho_{\Theta}(\SN{\psi_{1}, \psi_{2}}_{\Theta}).$$
That is, the action still provides a representation although we have lost the Jacobi identity for the `bracket algebra'.
\end{example}
\begin{example}\label{exm:contravariant}
In \cite{Asakawa:2015} it was shown how a  contravariant version of generalised geometry can be constructed starting from a Poisson manifold, their motivation was to construct a `stringy background' with R-flux. We can weaken the constructions slightly and obtain a `contravariant pre-Courant algebroid'. Let $(M, \mathcal{P})$ be a Poisson manifold.  Here we understand the Poisson structure as a function $\mathcal{P} \in \A^{2}(\Pi \sT^{*} M)$ such that $\SN{\mathcal{P}, \mathcal{P}}_{SN}=0$. The bracket here is the canonical Schouten--Nijenhuis bracket. We then pass to the Hamiltonian vector field associated with the Poisson structure and then consider  the symbol (or Hamiltonian) of the resulting vector field. In short,  we have
\begin{align*}
& \mathcal{P} = \frac{1}{2}\Lambda^{ab}(x) x^{*}_{b}x^{*}_{a}  && \mapsto  &  \theta := \Lambda^{ab}(x)x^{*}_{b}p_{a} - \frac{1}{2}\pi^{c}\frac{\partial \Lambda^{ab}}{\partial x^{c}}(x) x^{*}_{b} x^{*}_{a} \in \A^{3}(\sT^{*}\Pi \sT^{*}M),
\end{align*}
where we have employed local coordinates $(x^{a}, x^{*}_{b}, p_{c}, \pi^{d})$ of total weight $0$, $1$, $2$ and $1$ respectively on $\sT^{*} \Pi \sT^{*}M$. Note that $\theta$ is the symbol of the tangent lift $\xd_\sT\mathcal{P}$ \emph{modulo} the Tulczyjew isomorphism $\sT^{*} \Pi \sT M\simeq\sT^{*} \Pi \sT^{*}M$ (see  \cite{Grabowski:1995}), and that $\{\theta, \theta \} =0$ is equivalent to $\mathcal{P}$ being a Poisson structure, i.e., $\SN{\mathcal{P}, \mathcal{P}}_{SN}=0$. Thus, we have a symplectic Lie 2-algebroid and so a genuine Courant algebroid. \par
Now we add  `R-flux' $\mathcal{R} \in \A^{3}(\Pi \sT^{*}M)$ to the mix. We do not assume any compatibility condition between this `R-flux' and the Poisson structure,  usually it is assumed that $\SN{\mathcal{P} , \mathcal{R}}_{SN} =0$. We then define a  pre-Courant algebroid structure via the pre-homological potential
$$\Theta_{\mathcal{R}} := \theta + \mathcal{R} =  \Lambda^{ab}(x)x^{*}_{b}p_{a} - \frac{1}{2} \pi^{c}\frac{\partial \Lambda^{ab}}{\partial x^{c}}(x) x^{*}_{b} x^{*}_{a} + \frac{1}{3!} R^{abc}(x)x^{*}_{c}x^{*}_{b} x^{*}_{a}.$$
Note that
$$\{\{\theta, \mathcal{R}  \},f  \} =0,$$
for all $f \in C^{\infty}(M)$, even if we do not have $\{\theta, \mathcal{R}  \} =0$, as there are no conjugate variables present in the final expression. Similar reasoning shows that $\{ \{\mathcal{R}, \mathcal{R}  \}, f\}=0$. Thus,
$$\{ \{ \Theta_\mathcal{R}, \Theta_\mathcal{R} \}, f \} = \{\{ \theta, \theta\}, f \} + 2 \{\{ \theta, \mathcal{R}\},f \} + \{ \{ \mathcal{R}, \mathcal{R} \}, f \}=0,$$
\noindent for all $f \in C^\infty(M)$. We thus has the structure of a pre-Courant algebroid.
\end{example}
\begin{coexample}\label{exm:quasiP}
Modifying Example \ref{exm:contravariant} slightly, let us start with  a quasi-Poisson structure, that is a  function $\mathcal{P} \in \A^{2}(\Pi \sT^{*} M)$, but now  $\SN{\mathcal{P}, \mathcal{P}}_{SN} \neq 0$. As before we have
$$\Theta := \Lambda^{ab}(x)x^{*}_{b}p_{a} - \frac{1}{2}\pi^{c}\frac{\partial \Lambda^{ab}}{\partial x^{c}}(x) x^{*}_{b} x^{*}_{a} \in \A^{3}(\sT^{*}\Pi \sT^{*}M).$$
Clearly, we do not have a Courant algebroid structure, but do we have a pre-Courant algebroid structure? The answer is no. A direct calculation shows that
$$\{\Theta, \{ \Theta, f\} \} = \frac{1}{2}\left(\Lambda^{bc}\frac{\partial \Lambda^{ad}}{\partial x^{c}} -  \Lambda^{dc}\frac{\partial \Lambda^{ab}}{\partial x^{c}} - \Lambda^{ac}\frac{\partial \Lambda^{bd}}{\partial x^{c}} \right)x^{*}_{d}x^{*}_{b}\frac{\partial f}{\partial x^{a}},$$
for any $f \in C^{\infty}(M)$. Thus, we  must require $\mathcal{P} \in \A^{2}(\Pi \sT^{*} M)$ to be Poisson if the above is to vanish, and so $\{ \Theta, \Theta\} =0$. That is, given a quasi-Poisson structure we can build a symplectic nearly Lie 2-algebroid, but \emph{not} a pre-Courant algebroid, i.e., a symplectic almost Lie 2-algebroid.
\end{coexample}

\subsection{Sub-Dirac structures and almost Lie algebroids}
Example \ref{exm:bialgebroid} leads us to the following observation.
\begin{proposition}
The cotangent bundle of any almost Lie algebroid $E$ is canonically a pre-Courant algebroid.
\end{proposition}
\begin{proof}
This follows from Example \ref{exm:bialgebroid}  by equipping the dual  vector bundle $ E^{*}$ with the trivial Lie algebroid structure. That is, we consider  the Hamiltonian on $\sT^* \Pi E$ with $\theta \neq 0$ and  $\nu =0$.
\end{proof}
The  above proposition and the previous examples show that pre-Courant algebroids are to almost Lie algebroids what Courant algebroids are to Lie algebroids. The converse,  via Dirac structures, is also true. Our  general understanding of Dirac structures comes from  \v{S}evera \cite{Severa:2005,Severa:2017}.
\begin{definition}\label{def:Dirac}
Let $(F_{2}, \rmh, \{, \}, \Theta)$ be a symplectic almost Lie $2$-algebroid. A \emph{sub-Dirac structure} is a graded subbundle $\iota : \mathcal{D} \hookrightarrow F_{2}$ over $M_{0} := \rmh_{0}(\mathcal{D})$ such that:
\begin{enumerate}
\item $\mathcal{D}$ is isotropic;
\item $Q_{\Theta}$ is tangent to $\mathcal{D}$.
\end{enumerate}
Moreover, if the isotropic subbundle is maximal, then we speak of a \emph{Dirac structure}.
\end{definition}
\begin{proposition}\label{prop:Dirac}
If $\mathcal{D}$ is a sub-Dirac structure, then the vector bundle $L$ defined by $\Sec(L) \simeq  \mathcal{A}^{1}(\mathcal{D})$ is an almost Lie algebroid.
\end{proposition}
\begin{proof}
 The condition that $Q_{\Theta}$ is tangent to $\mathcal{D}$ ensures that $\mathcal{A}^{1}(\mathcal{D})$ is closed under the Courant--Dorfman pre-bracket. Because $\mathcal{D}$ is an isotropic subbundle, restriction of the pre-bracket to sections of $\mathcal{D}$ is skew-symmetric, see \eqref{eqn:symmetry}. Taking into account the  shift in Grassmann parity of elements of $\mathcal{A}^{1}(\mathcal{D})$, we see that the axioms of Definition \ref{def:almostLie} are satisfied.
\end{proof}
\begin{remark}
Proposition \ref{prop:Dirac} was first discovered by Gr\"{u}tzmann \&  Xu \cite[page 4]{Grutzmann:2012} in the context of sub-Dirac structure supported over \emph{all} of $M$. Sheng \& Liu  \cite{Sheng:2013} showed that for H-twisted Courant algebroids,  Dirac structures are  H-twisted Lie algebroids (also see \cite{Grutzmann:2011,Hansen:2010}).
\end{remark}
 \begin{example}[Adapted from \cite{Grabowski:2014}]
 Consider the supermanifold $\sT^{*}\Pi E$ which we equip  with local coordinates
$$(\underbrace{x^{a}}_{(0,0)},~ \underbrace{\zx^{i}}_{(1,0)},~\underbrace{p_{b}}_{(1,1)}, ~ \underbrace{\pi_{j}}_{(0,1)}),$$
where we have used the phase lift of the natural degree one homogeneity structure on $\Pi E$ to define the graded structure on the cotangent bundle.  We can then pass to total weight and clearly, we have a degree $2$ symplectic graded bundle, we will assume that we have an N-manifold and so $E$ is a vector bundle in the category of smooth manifolds.  Furthermore, let us assume that we have a \emph{projectable} pre-homological potential that defines a pre-Courant algebroid structure. By projectable we mean that the odd vector field $Q_{\Theta} = \{\Theta, - \}$ is projectable to $\Pi E$. The projectability directly implies that the pre-homological potential is of the local form
$$\Theta =  \zx^{i}Q_{i}^{a}(x)p_{a} + \frac{1}{2!}\zx^{i}\zx^{j}Q_{ji}^{k}(x)\pi_{k} + \frac{1}{3!} \zx^{i}\zx^{j} \zx^{k}Q_{kji}(x) := \Theta_{(2,1)} + \Theta_{(3,0)}.$$
Clearly,  $\Pi E\subset\sT^*\Pi E$ is a Dirac structure  and $E$ is an almost Lie algebroid via Proposition \ref{prop:almostLieQ}.   It is obvious that $\Theta = 0$ on $\Pi E^*$. \par
Now let us look for another Dirac structure in $\sT^{*}\Pi E$. Consider a bi-vector field $\Lambda = \frac{1}{2}\Lambda^{ij}(x)\pi_{j}\pi_{i}$ which we view as a function on the supermanifold $\Pi E^{*}$. Associated with any bi-vector is the Lagrangian submanifold $\iota_\Lambda : \mathcal{D}_{\Lambda} \hookrightarrow  \sT^{*}\Pi E$ given in local coordinates as
\begin{align*}
& \zx^{i} \circ \iota_\Lambda = \{ \Lambda, \zx^{i} \}, && p_{a} \circ \iota_\Lambda = \{\Lambda , p_{a}  \}.
\end{align*}
Direct calculation (which is not illuminating) gives
$$\Theta \circ \iota_\Lambda =  \frac{1}{2} \{ \Lambda, \{ \Lambda , \Theta_{(2,1)} \} \} + \frac{1}{3!} \{ \Lambda, \{ \Lambda, \{ \Lambda, \Theta_{(3,0)}\} \} \}.$$
As $\mathcal{D}_{\Lambda}$ is a Dirac structure if and only if $\Theta \circ \iota_\Lambda =0$, we see that  $\mathcal{D}_{\Lambda}$ is a Dirac structure if and only if $\Lambda$ is a  \emph{twisted Poisson structure}  on the almost Lie algebroid $E$ (see \cite{Kosmann-Schwarzbach:2005} and references therein).
 \end{example}
 \begin{example} In relation to the previous example, let us now consider the graph of a `two from' $\alpha = \frac{1}{2!} \zx^{i}\zx^{j}\alpha_{ji}(x)$, which we consider as a function on $\Pi E$. Associated with any two form is the Lagrangian submanifold $\iota_\alpha : \mathcal{D}_{\alpha} \hookrightarrow \sT^{*} \Pi E$ given in local coordinates as
 \begin{align*}
& p_{a} \circ \iota_\alpha = \frac{\partial \alpha}{\partial x^{a}}, && \pi_{i} \circ \iota_\alpha = \frac{\partial \alpha}{\partial \zx^{i}}.
\end{align*}
Thus, we see that
$$\Theta \circ \iota_\alpha =  \zx^{i}Q_{i}^{a}\frac{\partial \alpha}{\partial x^{a}} + \frac{1}{2!} \zx^{i}\zx^{j}Q_{ji}^{k}\frac{\partial \alpha}{ \partial \zx^{k}} + \frac{1}{3!}\zx^{i} \zx^{j} \zx^{k}Q_{kji} \,,$$
 so $\mathcal{D}_\za$ is a Dirac structure if and only if $\rmd_{E}\alpha + \Theta_{(3,0)} =0$.
 \end{example}

\subsection{The standard  quasi-cochain complex}
In light of Theorem \ref{thm:preCourant}, we see that given any pre-Courant algebroid we can construct a   quasi-cochain complex.
\begin{definition}
 The \emph{standard quasi-cochain complex} of a pre-Courant is defined as the quasi-cochain complex
$$\left(\A^{\bullet}(F_{2}), \: Q_{\Theta} = \{\Theta, - \}\right).$$
\end{definition}
The standard quasi-cochain complex is  \emph{not}, in general, a cochain complex as $Q^{2}_{\Theta}\neq 0$.  The above definition also holds equally as well for more general `non-homological Courant algebroids' (i.e., symplectic \emph{nearly} Lie 2-algebroids), but we will restrict our attention to pre-Courant algebroids. When we are dealing with genuine Courant algebroids  we have the \emph{standard cochain complex} as defined by Roytenberg \cite{Roytenberg:2002}. \par
 However, pre-Courant algebroids do have a well defined zeroth and first standard cohomology groups, as $Q_{\Theta}^{2}(f) =0$ for  all $f \in \A^{0}(F_{2}) \simeq C^{\infty}(M)$. These cohomology groups have the same interpretation as the standard cohomology of a Courant algebroid (cf. \cite{Roytenberg:2002}).   The zeroth standard cohomology group $H_{st}^{0}(E)$ consists of functions on $M$ that are constant along the leaves of the characteristic foliation, which is  defined  as the image of the anchor. The first standard cohomology group $H_{st}^{1}(E)$ consists of sections of $E$ that are in the kernel of the anchor, modulo $Q_\Theta$-exact sections. The higher cohomology groups are  not defined. \par
Recall that  associated with any pre-Courant algebroid we also have the \emph{na\"{\i}ve quasi-cochain complex}, see Definition \ref{def:naivequasicomplex}.  In order to reformulate the na\"{\i}ve quasi-cochain complex in terms of  symplectic nearly Lie 2-algebroids, we  first define the map
\begin{align*}
& \iota^\Theta : C^{\infty}(M) \times \A^{k}(F_{2}) \rightarrow \A^{k-1}(F_{2}),\\
& (f, X) \mapsto \iota^\Theta_f X := (-1)^{\widetilde{f}}\{ \{\Theta, f\}, X \}.
\end{align*}
We then define
$$\cL^{k}(F_{2}) := \left \{ \alpha \in \A^{k}(F_{1}) \subset \A^{k}(F_{2})~ | ~~ \iota^\Theta_f\alpha =0 ~ \textnormal{for all} ~ f \in C^{\infty}(M)  \right \}.$$
Note that $\cL^{1}(F_{2})$ consists of all  $\kappa \in \A^{1}(F_{1})$ such that $\{ \{\Theta, \kappa \}, f \} =0$ for all $f$. Thus, for $k=1$ we simply have the kernel of the anchor.
\begin{proposition}
Given any symplectic almost Lie 2-algebroid  $(F_2,\{\cdot,\cdot\},\Theta)$, the pair $(\cL^{\bullet}(F_{2}),~ Q_{\Theta})$ is a quasi-cochain complex.
\end{proposition}
\begin{proof}
We just need to show that $\cL^{\bullet}(F_{2})$ is closed with respect to $Q_{\Theta}$. Via direct application of the Jacobi identity, it is a straightforward exercise to show that
$$\iota^\Theta_{f}(Q_{\Theta}\alpha) + (-1)^{\widetilde{f}}Q_{\Theta}(\iota^\Theta_{f}\alpha) = \{ \{\Theta, \{ \Theta,f\}\}, \alpha \}.$$
Then we see that $\iota^\Theta$ and $Q_{\Theta}$ commute if and only if  we deal with a pre-Courant algebroid. Thus, if $\alpha$ is in $\cL^{\bullet}(F_{2})$ then so is $Q_{\Theta}\alpha$.
\end{proof}
\begin{remark}
It is important to note that we do \emph{not} obtain a quasi-cochain complex in this way for symplectic nearly Lie 2-algebroids.
\end{remark}
For the sake of simplicity and clarity, let us consider  pre-Courant algebroids in the category of purely even manifolds only. Much like the case of standard Courant algebroids we have the following theorem.
\begin{theorem}\label{thm:quasimap}
There exists a quasi-cochain map between the na\"{\i}ve quasi-cochain complex and the standard quasi-cochain complex associated with any pre-Courant algebroid for which the Grassmann parity and weight coincide (mod $2$).
\end{theorem}
\begin{proof}
Clearly, we have a module homomorphism between $\Sec(\wedge^{k} E^{*})$ and $\A^{k}(F_{1})$. Moreover the conditions $i_{\mathcal{D}f}\Psi =0$ and $\iota^\Theta_{f}\alpha=0$ (where $\alpha$ is  the image of $\Psi$ under the  aforementioned  module homomorphism) are equivalent and so we obtain a module homomorphism
$$\phi: C^{k}(E) \rightarrow \cL^{k}(F_{2}).$$
The only thing to prove is that $\mathcal{D}$ and $Q_{\Theta}$ are $\phi$-related. However, this follows from the same arguments given by  Sti\'{e}non \& Xu \cite{Stienon:2008} with only minor changes due to the fact that we use a non-skewsymmetric pre-bracket.  It is sufficient to consider the k=1 case and identify $\Psi \in \Sec(E^{\ast})$ as an element of $\Sec(E)$ using the pseudo-Euclidean structure. The general case follows from the Leibniz rule and linearity.  Now, it is a simple exercise to show that
\begin{align*}
\langle \Psi \circ e_{1}, e_{2} \rangle -  \mathcal{D}\Psi(e_{1},e_{2}) & = \rho(e_{2}) \langle \Psi, e_{1}\rangle -  \langle e_{2}\circ \Psi, e_{1}\rangle - \langle \Psi, e_{2}\circ e_{1}\rangle\\
&  + \langle \mathcal{D}\langle e_{2}, \Psi\rangle, e_{1} \rangle - \rho(e_{1})\langle \Psi, e_{1}\rangle + \langle \Psi, \mathcal{D}\langle e_{1}, e_{2}\rangle \rangle,
\end{align*}
where we have used $e_{1}\circ e_{2} + e_{2}\circ e_{1} = \mathcal{D}\langle e_{1},e_{2} \rangle $. Then, using the definition of $i^\Theta_{\mathcal{D}f} \Psi$, we obtain
$$\langle \Psi \circ e_{1}, e_{2} \rangle -  \mathcal{D}\Psi(e_{1},e_{2})  = \rho(e_{2}) \langle \Psi, e_{1}\rangle -  \langle e_{2}\circ \Psi, e_{1}\rangle - \langle \Psi, e_{2}\circ e_{1}\rangle + \rho(\Psi)\langle e_{1}, e_{2} \rangle.$$
 As $\Psi \in C^{1}(E)$, the final term of the above vanishes. Note that under the module homomorphism $\phi$ we have the identification (with minor abuse of notation)
$$ \langle \Psi \circ e_{1}, e_{2} \rangle \simeq  \{\{\{\Theta, \Psi \},e_{1}\},e_{2}\}.$$
This we can identify the two quasi-differentials, i.e.,
$$\mathcal{D}\Psi   \simeq \{\Theta, \Psi \},$$
and thus we have a quasi-cochain map.
\end{proof}
From the expression for the Jacobiator \eqref{eqn:Jacobi}, it is clear that if
 \begin{equation}\label{eqn:VanishingJac}
 J_{\Theta}(\kappa, -, -)=0,
 \end{equation}
  for all $\kappa \in \cL^{1}(F_{2}) \subset \A^{1}(F_{1})$, then $(\cL^{\bullet}(F_{2}), Q_{\Theta})$ is a cochain complex. The resulting cohomology, following  \cite{Liu:2016}, we refer to as the \emph{na\"{\i}ve cohomology}. Note that we do not, in general, require that $Q_{\Theta}^{2}=0$ on the whole of $\A^{\bullet}(F_{2})$, and so we are not forced to consider Courant algebroids.\par
  The condition on the Jacobiator \eqref{eqn:VanishingJac} is quite restrictive, but not prohibitively so. In particular,  twisted Courant algebroids (see \cite{Hansen:2010}) are natural examples of pre-Courant algebroids that have a well-defined na\"{\i}ve cohomology.


\section{Pre-Courant algebroids with an additional grading }\label{sec:NGrading}
\subsection{Weighted pre-Courant algebroids}
Following the  definition of a weighted (almost) Lie algebroid (Definition \ref{def:weightedalmostLie}), the notion of a  weighted pre-Courant algebroid is intuitively  clear.  In order to encode this structure in an economical way, we take the  following definition.
\begin{definition}\label{def:weightedPreC}
A \emph{weighted pre-Courant algebroid} of degree $k$ is a symplectic almost Lie 2-algebroid equipped with an additional compatible homogeneity structure of degree $k-1$.
\end{definition}
\noindent To be more explicit, we have:
\begin{enumerate}
\item A double graded (super) bundle $(F :=  F_{(k-1,2)} , \:  \rml, \:\rmh)$, where the homogeneity structures are of degree $k-1$ and $2$ respectively.
\item A nondegenerate Poisson bracket $\{\cdot,\cdot\}$ on the double graded bundle of bi-weight $(1-k, {-}2)$.
\item A pre-homological potential $\Theta \in \A^{(k-1,3)}(F)$, such that $\{\{\Theta, \Theta\}, f\} =0$ for all $f \in \A^{(\bullet , 0)}(F)$.
\end{enumerate}
The construction of the Courant--Dorfman pre-bracket, the anchor map and the pairing remains formally the same as before (see Subsection \ref{subsec:SymplecticAlmost}). However, we now understand $\A^{(\bullet, 1)}(F)$ to be closed under the Courant--Dorfman pre-bracket. Note that $\A^{(\bullet, 1)}(F)$ has a module structure over $\A^{(\bullet , 0)}(F)$, and thus we can identify $\A^{(\bullet, 1)}(F)$ with the sections of the graded vector bundle $\pi :F_{(k-1,1)} \rightarrow F_{(k-1, 0)}$, identified with its dual with the use of the pseudo-Euclidean structure. Just by counting the weights, we see that
\begin{equation}\label{eqn:weightedClosure}
\SN{\A^{(p,1)}(F) , \A^{(q,1)}(F)}_{\Theta} \subset \A^{(m,1)}(F),
\end{equation}
where $m = p+q-k+1$. The anchor can be considered as a  map
$$\rho_{\Theta} : \A^{(p,1)}(F) \rightarrow \Vect(F_{(k-1,0)}),$$
where  $\rho_\Theta(\sigma)$ is homogeneous and of weight $p+1-k$.  The pseudo-Euclidean structure can be considered as   a map
$$\langle \: , \:  \rangle : \A^{(p,1)}(F) \times  \A^{(q,1)}(F) \rightarrow \A^{(m,0)}(F) ,$$
where  again $m = p+q-k +1$.  Note that if $p+q-k <1$, then both the Courant--Dorfman pre-bracket and the pairing take the value zero. Extension to inhomogeneous sections is via linearity.\par
We can always employ affine Darboux coordinates $\big(x^{A}_{w}, p_{B}^{w}, \zx^{I}_{w}\big)$ on $F_{(k-1,2)}$. Notationally, we understand the additional label to represent the weight with respect to the homogeneity structure $\rml : \mathbb{R} \times F \rightarrow F$, thus $0 \leq w < k$. Because we have employed the phase lift, $x^{A}_{w}$ and $p_{B}^{k-w-1}$ are conjugate coordinates. It is clear that the pre-homological potential must be of the form
$$\Theta = \zx^{I}_{w}Q_{I}^{A}[w'-w](x)p_{A}^{k-w'-1} +\frac{1}{3!} \zx^{I}_{w}\zx^{J}_{w'}\zx^{K}_{w''}Q_{KJI}[k-1-w-w'-w''](x).$$
Here $Q^{A}_{I}[\bullet]$  and $Q_{KJI}[\bullet]$ are the homogeneous pieces of the structure functions. Any expression that is seemingly negative in weight is understood to be zero.  Via inspection, we arrive at the following proposition.
\begin{proposition}\label{prop:Wanchor}
The anchor map of a weighted pre-Courant algebroid can be considered as a morphism in the category of double graded (super)bundles
$$\rho_{\Theta} : F_{(k-1,1)} \rightarrow \Pi \sT F_{(k-1,0)},$$
where the homogeneity structures on the antitangent bundle arise from the tangent lift of the  homogeneity structure on $F_{(k-1,0)}$  and the natural regular homogeneity structure associated with any tangent bundle.
\end{proposition}
The degree $1$ case is just the theory of symplectic almost Lie 2-algebroids, and so via Theorem \ref{thm:preCourant} covers the standard theory of pre-Courant algebroids. The degree $2$ case covers  VB-Courant algebroids (cf. \cite{Li-Bland:2012}).
\begin{theorem}\label{thm:highertangentlift}
Let $(F_{2}, \rmh,  \{~ ,~\}, \Theta)$ be a pre-Courant algebroid. Then the higher tangent bundle $\sT^{k-1}F_{2}$ is canonically a weighted pre-Courant algebroid of degree $k$.
\end{theorem}
\begin{proof}
Clearly, $\sT^{k-1} F_{2}$ is a double graded bundle where the homogeneity structures are the natural one on a higher tangent bundle and the prolongation of the homogeneity structure on $F_{2}$ (see Example \ref{exm:highertangent}). Thus we can employ natural local coordinates
$$(\underbrace{x^{a}_{\epsilon}}_{(\epsilon, 0)}, \: \underbrace{\zx^{i}_{\epsilon}}_{(\epsilon, 1)}, \: \underbrace{p_{b}^{\epsilon}}_{(\epsilon , 2)}),$$
on $\sT^{k-1}F_{2}$. Here $0\leq \epsilon < k$ denotes the weight which we understand as coming from the `number of derivatives' of the original coordinates on $F_{2}$.  The Poisson bracket also lifts (see \cite{KouotchopWamba:2012}) with the conjugate variables being $(x_{\epsilon}^{a}, p_{b}^{\delta})$ and $(\zx^{i}_{\epsilon}, \zx^{j}_{\delta})$, where $\delta = \epsilon -k+1$. Thus we have a non-degenerate Poisson bracket of bi-weight $(1-k,-2)$. The higher tangent prolongation  (or complete lift cf. \cite{Morimoto:1970}) of $\Theta$ to a function on $\sT^{k-1}F_{2}$, which we denote as $\Theta^{(c)}$, is formally constructed in the same way as the classical case on purely even manifolds  by taking derivatives. It is not hard to see that $\Theta^{(c)}$ is of bi-weight $(k-1,3)$. Moreover, it follows that (again see \cite{KouotchopWamba:2012})
$$\{\Theta^{(c)}, \Theta^{(c)} \}_{\sT^{k-1}F_{2}} =  \{\Theta, \Theta  \}^{(c)}.$$
Note that $\{ \{\Theta, \Theta \} \}, f_{0}\} =0$ for any $f_{0} \in C^{\infty}(M)$ implies that $\{ \Theta, \Theta\}$ is independent of the momenta: this is equivalent to the structure equations of the pre-Courant algebroid. Thus, we deduce that $\{ \Theta^{c} , \Theta^{c} \}_{\sT^{k-1}F_{2}}$ is similarly independent of the coordinates $p^{\epsilon}$,  that is we have a function on $\sT^{k-1} F_{1}$. We then see that
$$\{ \{\Theta^{(c)}, \Theta^{(c)}\}_{\sT^{k-1}F_{2}} , f \}_{\sT^{k-1}F_{2}} =0,$$
for all $f \in \mathcal{A}(\sT^{k-1}M)$, due to the lack of conjugate variables and conclude that we  have the structure of a weighted pre-Courant algebroid.
\end{proof}
Two immediate consequences of Theorem \ref{thm:highertangentlift} are the following.

 \begin{corollary}
\
 \begin{enumerate}
 \item The higher tangent bundle of order $k-1$ of a Courant algebroid is a weighted Courant algebroid of degree $k$;
 \item The tangent bundle of a (pre-)Courant algebroid is a weighted (pre-)Courant algebroid of degree $2$.
 \end{enumerate}
 \end{corollary}
 \begin{remark}
 As far as we are aware, the first proof that  the tangent bundle of a Courant algebroid canonically comes with the structure of a VB-Courant algebroid was given by  Boumaiza \& ~Zaalani \cite{Boumaiza:2009} in the standard classical  framework. Li-Bland  in his Ph.D. thesis \cite{Li-Bland:2012} sketched a proof of this in the language of supermanifolds.
\end{remark}

\subsection{VB-Courant algebroids revisited}
As a degree $1$ (and so regular) homogeneity structure encodes a vector bundle structure (see \cite{Grabowski:2009}), the natural definition of a  VB-pre-Courant algebroid is as follows.
\begin{definition}\label{def:VBpreC}
A \emph{VB-pre-Courant algebroid} is a weighted pre-Courant algebroid  of degree $k=2$ (see Definition \ref{def:weightedPreC}).
\end{definition}
 Restricting our attention to the $k=2$ case, within $\A^{(\bullet,1)}(F)$ there are two privileged modules over $\A^{(0,0)}(F)$,  $\A^{(1,1)}(F)$  and $\A^{(0,1)}(F)$. In the standard formulation of double vector bundles, $\A^{(1,1)}(F)$ corresponds to   \emph{linear sections},  and $\A^{(0,1)}(F)$ corresponds to \emph{core sections} of the dual of the double vector bundle $F_{(1,1)}$ over $F_{(1,0)}$. By counting the weights (see \eqref{eqn:weightedClosure}) we observe  that
 \begin{align*}
 & \SN{\A^{(1,1)}(F), \A^{(1,1)}(F)}_{\Theta} \subset \A^{(1,1)}(F), && \SN{\A^{(1,1)}(F), \A^{(0,1)}(F)}_{\Theta} \subset \A^{(0,1)}(F), & \SN{\A^{(0,1)}(F), \A^{(0,1)}(F)}_{\Theta} =0.
 \end{align*}
 Thus, the Courant--Dorfman pre-bracket is linear. Moreover, Proposition \ref{prop:Wanchor} implies that the anchor, thought of as a map  $F_{(1,1)} \rightarrow \sT F_{(1,0)}$, is also linear. One can also directly observe that
 \begin{align*}
 & \langle \A^{(1,1)}(F), \A^{(1,1)}(F) \rangle \subset \A^{(1,0)}(F), && \langle \A^{(1,1)}(F), \A^{(0,1)}(F) \rangle \subset \A^{(0,0)}(F), & \langle \A^{(0,1)}(F), \A^{(0,1)}(F) \rangle =0,
 \end{align*}
 which by Gracia-Saz \& Mehta \cite[Lemma 2.8]{Gracia-Saz:2009} implies that the pairing is also linear.\par
Thus our definition of a VB-pre-Courant algebroid gives rise to what one would expect in a more traditional language by weakening the classical axioms of a VB-Courant algebroid, see Li-Bland \cite{Li-Bland:2012}. In particular,  \cite[Proposition 3.2.1]{Li-Bland:2012}  essentially leads to our \emph{definition} of a VB-(pre-)Courant algebroid, once homogeneity structures are employed to describe vector bundles.

\subsection{Weighted sub-Dirac structures}
Following Definition \ref{def:Dirac}, the notion of a weighted sub-Dirac structure is clear in terms of a sub-Dirac structure with a compatible homogeneity structure. More formally, we have the following definition.
\begin{definition}
Let $(F_{(k-1,2)} , \rml , \rmh, \{~,~\} , \Theta)$ be a weighted pre-Courant algebroid of degree $k$. A \emph{weighted sub-Dirac structure} of degree $k$ is a double graded subbundle $\iota : \mathcal{D}_{(k-1,2)}\hookrightarrow F_{(k-1,2)}$, possibly supported over $F_{0 (k-1,0)} :=  \rmh_{0} (D_{(k-1,2)})$ and $F_{0(0,2)} := \rml_{0} (D_{(k-1,2)})$,
\begin{center}
\leavevmode
\begin{xy}
(0,15)*+{\mathcal{D}_{(k-1,2)}}="a"; (25,15)*+{F_{0(0,2)}}="b";%
(0,0)*+{F_{0(k-1,0)}}="c"; (25,0)*+{M_{0}}="d";%
{\ar "a";"b"}?*!/^3mm/{};%
{\ar "a";"c"}?*!/^3mm/{};{\ar "b";"d"}?*!/_3mm/{};%
{\ar "c";"d"} ?*!/^3mm/{};%
(55,15)*+{F_{(k-1,2)}}="e"; (80,15)*+{F_{(0,2)}}="f";%
(55,0)*+{F_{(k-1,0)}}="g"; (80,0)*+{M}="h";%
{\ar "e";"f"}?*!/^3mm/{};%
{\ar "e";"g"}?*!/^3mm/{};{\ar "f";"h"}?*!/_3mm/{};%
{\ar "g";"h"} ?*!/^3mm/{};%
(30,7,5)*+{}="i"; (50,7.5)*+{}="j";%
{\ar@{^{(}->} "i";"j"} ?*!/^3mm/{};%
\end{xy}
\end{center}
\smallskip
such that:
\begin{enumerate}
\item $\mathcal{D}_{(k-1,2)}$  is isotropic;
\item $ Q_{\Theta}$ is tangent to $\mathcal{D}_{(k-1,2)}$.
\end{enumerate}
If $\mathcal{D}_{(k-1,2)}$ is a  Lagrangian submanifold, then we speak of a \emph{weighted Dirac structure}.
\end{definition}
\begin{remark}
The notion of a sub-VB-Dirac structure is clear.  VB-Dirac structures, to our knowledge, first appear in \cite{Li-Bland:2012} in the more traditional framework of double vector bundles.
\end{remark}
\begin{proposition}
The graded vector bundle $L_{(k-1,1)}$ defined by $\Sec(L_{(k-1,1)}) \simeq \mathcal{A}^{(\bullet, 1)}(\mathcal{D}_{(k-1,2)})$ is a weighted almost Lie algebroid of degree $k$.
\end{proposition}
\begin{proof}
It is clear from Proposition \ref{prop:Dirac} that by forgetting the additional graded structure we get  an underlying almost Lie algebroid structure. The only question is the weight of the bracket on sections. A homogeneous section of weight $r$ of $L_{(k-1,1)}$ can be viewed as an element of $\mathcal{A}^{(r-1,1)}(\mathcal{D}_{(k-1,2)})$, and thus by \eqref{eqn:weightedClosure} the almost Lie bracket carries weight $-k$. Then, via Definition \ref{def:weightedalmostLie}, we have the structure of a weighted almost Lie algebroid.
\end{proof}
\begin{proposition}\label{prop:liftsubDirac}
Let $\iota: \mathcal{D}  \hookrightarrow F_{2}$ be a sub-Dirac structure of a pre-Courant algebroid. Then
 $$\sT^{k-1}\iota:  \sT^{k-1}\mathcal{D} \hookrightarrow \sT^{k-1}F_{2}$$
  is a weighted sub-Dirac structure of degree $k$.
\end{proposition}
\begin{proof}
Due to the functorial properties of taking the higher tangent bundle, it is clear that $\sT^{k-1}\mathcal{D}$ is a graded double subbundle of $\sT^{k-1}F_{2}$.  As $\mathcal{D}$ is isotropic, then so is $\sT^{k-1}\mathcal{D}$ with respect to the lifted Poisson structure. Similarly, as  Hamiltonian vector fields of the complete lifts of Hamiltonians are the complete lifts of the Hamiltonian vector fields, $Q_{\Theta^{(c)}} = (Q_{\Theta})^{c}$ (see \cite{KouotchopWamba:2012}), we see that the $Q_{\Theta^{(c)}}$ is tangent to $\sT^{k-1}\mathcal{D}$.
\end{proof}
\begin{remark}
Higher order tangent lifts of Dirac structures in $\sT M \times_{M}\sT^{*}M$ were  first explored by Kouotchop Wamba, Ntyam \& Wouafo Kamga \cite{KouotchopWamba:2011} in the classical language. Also, see Kouotchop Wamba \& Ntyam \cite{KouotchopWamba:2013}.
\end{remark}

\subsection{Examples  of  weighted pre-Courant algebroids}
In this subsection, we present a few simple examples of  weighted pre-Courant algebroids and sub-Dirac structures. These examples come from minor modifications of the standard examples of Courant algebroids and Dirac structures.
\begin{example}
Using Theorem \ref{thm:highertangentlift} and Proposition \ref{prop:liftsubDirac}, we can see that any pre-Courant algebroid and a sub-Dirac structure thereof, gives rise to a weighted pre-Courant algebroid and a weighted sub-Dirac structure via the higher order tangent functor. As a specific example,  consider $\sT^{*}\Pi \sT M$, where for simplicity  we assume $M$ is a  standard manifold, which canonically is a Courant algebroid. The Poisson structure is just the canonical Poisson structure on the cotangent bundle.  We need only remember that we assign weight $2$ to the `momenta'.  The homological potential is just the symbol (or Hamiltonian) of de Rham differential on $M$.  Obviously, $\Pi \sT M$  is a Dirac structure with the canonical Lie algebroid structure. Then,  $\sT^{k-1} (\sT^{*} \Pi \sT M) \simeq \sT^{*}\Pi \sT(\sT^{k-1}M)$ is canonically a weighted Courant algebroid of degree $k$, and $\Pi \sT(\sT^{k-1}M)$ is a weighted Dirac structure thereof.
\end{example}
\begin{example}
Consider $\sT^{*} F_{k-1}$, where $F_{k-1}$ is any degree $k-1$ graded  bundle, and assign weight $2$ to the `momenta'. That is, we can employ local coordinates of the form
$$(\underbrace{x^{a}}_{(0,0)}, \:  \underbrace{y^{i}_{w}}_{(w,0)},\:  \underbrace{p_{a}^{k-1}}_{(k-1,2)},\: \underbrace{q_{i}^{w}}_{(w,2)}),$$
where $(0 < w <  k)$. Note that the conjugate pairs here are $(x,p^{k-1})$ and $(y_{w}, q^{k-1 -w})$, and so the canonical Poisson bracket on the cotangent bundle carries bi-weight $(1-k, -2)$. However, notice that there are no functions of bi-weight $(\bullet , \textnormal{odd})$ and so the only possible choice of a (pre-)homological potential is just $\Theta =0$.   Any isotropic subbundle  of $\sT^{*} F_{k-1}$ is a sub-Dirac structure.
\end{example}
\begin{example}
As an example of a VB-Courant algebroid structure, consider $F_{(1,2)} :=  \sT^{*}\Pi \sT E^{*}$, where $E \rightarrow M$ is a vector bundle, which we will for simplicity assume to be in the category of  manifolds. We equip $F_{(1,2)}$ with local coordinates of the form
$$(\underbrace{x^{a}}_{(0,0)}, \: \underbrace{q_{i}}_{(1,0)}, \: \underbrace{\rmd x^{b}}_{(0,1)}, \: \underbrace{\rmd q_{j}}_{(1,1)}, \:  \underbrace{p_{c}}_{(1,2)},\:  \underbrace{y^{k}}_{(0,2)}, \: \underbrace{\pi_{d}}_{(1,1)},\: \underbrace{\chi^{l}}_{(0,1)}).$$
With this assignment of bi-weights, it is clear that the canonical Poisson bracket carries bi-weight $(-1,-2)$ as desired. Canonically,
$$\theta =  \rmd x^{a} p_{a} + \rmd q_{i} y^{i}$$
provides a homological potential and we have the structure of a VB-Courant algebroid or in our language a weighted Courant algebroid of degree $2$.\par
Now, let us `twist' the homological potential by adding \emph{any} linear 3-form
$$\beta =  \frac{1}{3!} \rmd x^{a} \rmd x^{b} \rmd x^{c}\beta_{cba}^{i}(x)q_{i} + \frac{1}{2!} \rmd x^{a} \rmd x^{b}\beta_{ba}^{i}(x) \rmd q_{i},$$
which we understand as a function of bi-weight $(1,3)$ on $\Pi \sT E^{*}$. Now consider the pre-homological potential
$$\Theta :=  \theta +  \beta.$$
A direct  calculation shows that
$$\{ \Theta, \Theta \} = 2 \{ \theta, \beta \} = 2 \: \rmd \beta \neq 0\,,$$
where $\rmd$ is the de Rham differential on $\Pi \sT E^{*}$. We do \emph{not} at this stage assume that $\beta$ is closed.  It is not hard to see that
$$\{ \{\Theta, \Theta \}, f \} = 2\{\rmd \beta, f  \} =0$$
for all $f \in \mathcal{A}^{\bullet}(E^{*})$, due to the lack of conjugate variables in the final Poisson bracket. Thus, we have  a weighted  pre-Courant algebroid of degree $2$.\par
An obvious example of a  weighted Dirac structure here is $\Pi \sT E^{*}$, which comes canonically with the de Rham differential. Let us find a less obvious example of weighted Dirac structure by considering a linear two form
$$\alpha  = \frac{1}{2} \rmd x^{a}\rmd x^{b}\alpha_{ba}^{i}(x)q_{i} + \rmd x^{a}\alpha_{a}^{i}(x)\rmd q_{i},$$
which we understand as a function of bi-weight $(1,2)$ on $\Pi \sT E^{*}$. Associated with any such two form is the Lagrangian subbundle
$$ \iota_\alpha : \mathcal{D}_{\alpha} \hookrightarrow \sT^{*}\Pi \sT E^{*}\,,$$
defined locally viz
\begin{align*}
& p_{a} \circ \iota_\alpha =  \frac{\partial \alpha}{\partial x^{a}}, & y^{i} \circ \iota_\alpha =  \frac{\partial \alpha}{\partial q_{i}}, && \pi_{b} \circ \iota_\alpha =  \frac{\partial \alpha}{\partial \rmd x^{b}}, && \chi^{j} \circ \iota_\alpha =  \frac{\partial \alpha}{\partial \rmd q_{j}}.
\end{align*}
Thus, to ensure that $Q_{\Theta}$ is tangent to $\mathcal{D}_{\alpha}$ we require $\Theta \circ \iota_\alpha = 0$, and so we arrive at
$$\rmd \alpha  + \beta =0.$$
The above then implies the consistency condition  $\rmd \beta =0$ and so, for this particular weighted Dirac  structure, we are forced to consider a genuine Courant algebroid.
 \end{example}

\section{Concluding remarks}\label{sec:Con}
We have shown that pre-Courant algebroids are `really' symplectic almost Lie 2-algebroids, just as Courant algebroids are `really' symplectic Lie 2-algebroids. This shift in the fundamental starting place allows for  many of the known facts about pre-Courant algebroids to be neatly reformulated. In many cases, the proofs of various statements become significantly simpler than those presented using a classical framework. For example, the  notion of  (sub-)Dirac structures on pre-Courant algebroids as certain isotropic submanifolds is conceptually neat and clear. \par
Working with symplectic almost Lie 2-algebroids allows for a very  economical and powerful framework to understand pre-Courant algebroids with a compatible non-negative grading. Weighted pre-Courant algebroids directly generalise VB-Courant algebroids, while at the same time  their use drastically simplifies  working  with the latter. In particular, the notion of weighted sub-Dirac structures is almost obvious and naturally covers the case of VB-Dirac structures.\par
Pre-Courant algebroids and especially weighted pre-Courant algebroids have a rich and interesting supergeometrical structure which, in our opinion, deserves further study. We are convinced that the understanding of pre-Courant algebroids as symplectic almost Lie 2-algebroids will allow for further mathematical developments and applications in theoretical physics.

\section*{Acknowledgements}
We thank Yunhe Sheng for his interest and useful comments with regards to this work.  Furthermore, we cordially thank the anonymous referee for their comments and suggestions which have served to vastly improve the overall presentation of this paper.


\begin{thebibliography}{10}
\begin{small}


\bibitem{Alexandrov:1997}
M.~Alexandrov, M.~Kontsevich, A.~Schwarz \&  O.~Zaboronsky,
\newblock{The geometry of the master equation and topological quantum field theory,}
\newblock{\emph{Internat. J. Modern Phys. A}  12 (1997), 1405--1429.}

\bibitem{Armstrong:2007}
S.~Armstrong,
\newblock{Note on pre-Courant algebroid structures for parabolic geometries,}
\newblock{\texttt{arXiv:0709.0919 [math.DG]} (2007).}

\bibitem{Armstrong:2011}
S.~Armstrong \& R.~Lu,
\newblock{Courant Algebroids in Parabolic Geometry,}
\newblock{\texttt{arXiv:1112.6425 [math-ph]} (2011).}

\bibitem{Asakawa:2012}
T.~Asakawa, S.~Sasa \& S.~Watamura,
\newblock{ D-branes in Generalized Geometry and Dirac-Born-Infeld Action,}
\newblock{\emph{JHEP} \textbf{1210} (2012), 064.}

\bibitem{Asakawa:2015}
T.~Asakawa,  H.~Muraki, S.~Sasa \& S.~Watamura,
\newblock{Poisson-generalised geometry and R-flux,}
\newblock{\emph{Internat. J. Modern Phys. A} 30 (2015), 1550097.}


\bibitem{Bessho:2016}
 T.~Bessho, M.A.~Heller, N.~Ikeda \& S.~Watamura,
 \newblock{Topological membranes, current algebras and H-flux-R-flux duality based on Courant algebroids,}
 \newblock{\emph{JHEP} \textbf{1604} (2016), 170.}

\bibitem{Bonavolonta:2013}
G.~Bonavolont\`{a} \& N.~Poncin,
\newblock{On the category of Lie n-algebroids,}
\newblock{ \emph{J. Geom. Phys.} \textbf{73} (2013), 70.}


\bibitem{Boumaiza:2009}
M.~Boumaiza \&  N.~Zaalani,
\newblock{ Rel\`{e}vement d'une alg\'{e}bro\"{\i}de de Courant}
 \newblock{Comptes Rendus Math\'{e}matique,  Acad\'{e}mie des Sciences, Paris, \textbf{347} (2009), 177--182.}

\bibitem{Bruce:2015}
A.J.~Bruce, K.~Grabowska \& J.~Grabowski,
\newblock{Graded bundles in the category of Lie groupoids,}
 \newblock{\emph{SIGMA} \textbf{11} (2015), 090.}

 \bibitem{Bruce:2016}
A.J.~Bruce, K.~Grabowska \& J.~Grabowski,
\newblock{Linear duals of graded bundles and higher analogues
of (Lie) algebroids,}
 \newblock{\emph{J. Geom. Phys.}  \textbf{101} (2016), 71--99.}

\bibitem{Courant:1990}
T.~Courant,
\newblock{Dirac manifolds,}
\newblock{\emph{Trans. Amer. Math. Soc.} \textbf{319} (1990), 631--661.}


\bibitem{Deser:2015}
A.~Deser \& J.~Stasheff,
\newblock{Even symplectic supermanifolds and double field theory,}
\newblock{\emph{Commun. Math. Phys.} \textbf{339} (2015), 1003--1020.}


\bibitem{Dorfman:1987}
 I.Ya.~Dorfman,
 \newblock{Dirac structures of integrable evolution equations,}
\newblock{\emph{Phys. Lett. A} \textbf{125} (1987), 240--246.}



\bibitem{Grabowska:2008}
K.~Grabowska \& J.~Grabowski,
\newblock{Variational calculus with constraints on general algebroids,}
\newblock{\emph{J. Phys. A}  \textbf{41} (2008), 175204,  25 pp.}


\bibitem{Grabowski:2013}
J.~Grabowski,
\newblock{Graded contact manifolds and contact Courant algebroids,}
\newblock{\emph{J. Geom. Phys.} \textbf{68} (2013), 27--58.}

\bibitem{Grabowski:2014}
J.~Grabowski,
\newblock{Modular classes revisited,}
\newblock{\emph{Int. J. Geom. Methods Mod. Phys.} \textbf{11}(9) (2014), 11p.}

\bibitem{Grabowski:2011}
\newblock J. Grabowski and M. J\'o\'zwikowski,
\newblock{Pontryagin Maximum Principle on almost Lie algebroids},
\newblock{ \emph{SIAM J. Control Optim.} \textbf{49} (2011), 1306--1357.}

\bibitem{Grabowski:2009}
J.~Grabowski \& M.~Rotkiewicz,
\newblock{Higher vector bundles and multi-graded symplectic manifolds,}
\newblock{\emph{J. Geom. Phys.} \textbf{59} (2009)  1285--1305.}

\bibitem{Grabowski:2012}
J.~Grabowski \& M.~Rotkiewicz,
\newblock{Graded bundles and homogeneity structures,}
\newblock{\emph{J. Geom. Phys.} \textbf{62} (2012), 21--36. }

\bibitem{Grabowski:1995} J.~Grabowski \& P.~Urba\'nski,
\newblock{Tangent lifts of Poisson and related structures,}
\newblock{\emph{J. Phys. A} \textbf{28} (1995), 6743--6777.}


\bibitem{Grabowski:1999}
J.~Grabowski \& P.~Urba\'{n}ski,
\newblock{Algebroids - general differential calculi on vector bundles,}
\newblock{\emph{J. Geom. Phys.}  \textbf{31} (1999)  111--141.}

\bibitem{Gracia-Saz:2009}
A.~Gracia-Saz \& R.A.~Mehta,
\newblock{Lie algebroid structures on double vector bundles and representation theory of Lie algebroids,}
\newblock{\emph{Adv. Math}. \textbf{223} (2010),  1236--1275.}

\bibitem{Grutzmann:2011}
M.~Gr\"{u}tzmann,
\newblock{H-twisted Lie algebroids,}
\newblock{\emph{J. Geom. Phys.} \textbf{61} (2011), 476-484.}

\bibitem{Grutzmann:2012}
M.~Gr\"{u}tzmann \&  X.~Xu,
\newblock{Cohomology for almost Lie algebroids,}
\newblock{\texttt{arXiv:1206.5466 [math.DG]} (2012).}

\bibitem{Hansen:2010}
M.~Hansen \& T.~Strobl,
\newblock{First class constrained systems and twisting of Courant algebroids by a closed 4-form,}
\newblock{in: Fundamental Interactions, A Memorial Volume for Wolfgang Kummer, World Scientific, 2010, pp. 115--144.}

\bibitem{Lang:2018}
H.~Lang, Y.~Sheng \& A.~Wade,
\newblock{VB-Courant algebroids, E-Courant algebroids and generalized geometry,}
\newblock{to appear in \emph{Canadian Mathematical Bulletin} (2018)}



\bibitem{Loday:1993}
J.L.~Loday,
\newblock{Une version non commutative des alg\`{e}bres de Lie: les alg\`{e}bres de Leibniz,}
\newblock{\emph{Enseign. Math.}, \textbf{39} (1993) pp. 269--293.}

\bibitem{JotzLean:2015}
M.~Jotz Lean,
\newblock{$N$-manifolds of degree $2$ and metric double vector bundles,}
\newblock{\texttt{arXiv:1504.00880 [math.DG]} (2015)}

\bibitem{Jozwikowski:2016}
M.~J\'{o}\'{z}wikowski \& M.~Rotkiewicz,
\newblock{A note on actions of some monoids,}
\newblock{\emph{Differential Geom. Appl.} \textbf{47} (2016), 212--245.}

\bibitem{Kosmann-Schwarzbach:1995}
Y.~Kosmann-Schwarzbach,
\newblock{Exact Gerstenhaber algebras and Lie bialgebroids,}
\newblock{ \emph{Acta Appl. Math}, \textbf{41}, 153-165 (1995).}

\bibitem{Kosmann-Schwarzbach:1996}
Y.~ Kosmann-Schwarzbach,
\newblock{From Poisson algebras to Gerstenhaber algebras,}
\newblock{\emph{Ann. Inst. Fourier} \textbf{46}(5) (1996), 1243--1274.}

\bibitem{Kosmann-Schwarzbach:2003}
Y.~ Kosmann-Schwarzbach,
\newblock{Derived brackets,}
\newblock{\emph{Lett. Math. Phys.} \textbf{69} (2004), 61--87.}

\bibitem{Kosmann-Schwarzbach:2005}
Y.~ Kosmann-Schwarzbach \& C.~ Laurent-Gengoux,
\newblock{The modular class of a twisted Poisson structure,}
\newblock{ \emph{Travaux math\'{e}matiques. Fasc. XVI}, 315--339,
\emph{Trav. Math.}, \textbf{16}, Univ. Luxemb., Luxembourg, 2005.}

\bibitem{Kosmann-Schwarzbach:2013}
Y.~ Kosmann-Schwarzbach,
\newblock{Courant Algebroids. A Short History,}
\newblock{ \emph{SIGMA} \textbf{9}(14) (2013), 8 pages.}

\bibitem{KouotchopWamba:2011}
P.M.~Kouotchop Wamba, A.~Ntyam \& J.~Wouafo Kamga,
\newblock{ Tangent Dirac structures of higher order,}
\newblock{ \emph{Arch. Math. (Brno)} \textbf{47} (2011), 17-22.}


\bibitem{KouotchopWamba:2012}
P.M.~Kouotchop Wamba, A.~Ntyam \& J.~Wouafo Kamga,
\newblock{ Tangent lifts of higher order of multivector fields and applications,}
\newblock{ \emph{J. Math. Sci. Adv. Appl.} \textbf{15}(2) (2012), 89--112.}

\bibitem{KouotchopWamba:2013}
P.M.~Kouotchop Wamba \& A.~Ntyam,
\newblock{Tangent lifts of higher order of multiplicative Dirac structures,}
\newblock{ \emph{Arch. Math. (Brno)} \textbf{49} (2013), 87-104.}

\bibitem{Li-Bland:2012}
D.S.~Li-Bland,
\newblock{LA-Courant Algebroids and their Applications,}
\newblock{PhD Thesis, University of Toronto (2012), 137 pages.}

\bibitem{Liu:2016}
Z.~Liu, Y.~Sheng \& X.~Xu,
\newblock{The Pontryagin class for pre-Courant algebroids,}
\newblock{\emph{J. Geom. Phys.} \textbf{104} (2016), 148--162.}

\bibitem{Liu:1997}
Z.~Liu, A.~Weinstein \& P.~Xu,
\newblock{Manin triples for Lie bialgebroids,}
\newblock{ \emph{J. Diff. Geom.} \textbf{45}(3) (1997), 547--574.}

\bibitem{Lyakhovich:2004}
S.L.~Lyakhovich \&  A.A.~Sharapov,
\newblock{Characteristic classes of gauge systems,}
\newblock{\emph{Nuclear Phys. B} \textbf{703}(30) (2004), 419--453. }

\bibitem{Mackenzie:1992}
K.C.H.~Mackenzie,
\newblock{Double Lie algebroids and second-order geometry. I,}
\newblock{\emph{Adv. Math.} \textbf{94}(2) (1992), 180--239.}

\bibitem{Mackenzie:2000}
K.C.H.~Mackenzie,
\newblock{Double Lie algebroids and second-order geometry. II,}
\newblock{\emph{Adv. Math.}  \textbf{154}(1) (2000),  46--75.}

\bibitem{Morimoto:1970}
A.~Morimoto,
\newblock{Lifting of some type of tensors fields and connections to tangent bundles
of $p^r$-velocities,}
\newblock{\emph{Nagoya Math. J.} \textbf{40} (1970), 13--31.}

\bibitem{Roytenberg:2002}
D.~Roytenberg,
\newblock{ On the structure of graded symplectic supermanifolds and Courant algebroids,}
\newblock{in: Quantization, Poisson brackets and beyond (Manchester, 2001), 169--185, \emph{Contemp. Math.} \textbf{315}, Amer. Math. Soc., Providence, RI, 2002.}

\bibitem{Roytenberg:2009}
D.~Roytenberg,
\newblock{AKSZ-BV Formalism and Courant Algebroid-induced Topological Field Theories,}
\newblock{ \emph{Lett. Math. Phys.} \textbf{79} (2009), 143-159.}

\bibitem{Severa:2005}
P.~\v{S}evera,
\newblock{Some title containing the words ``homotopy'' and ``symplectic'', e.g. this one,}
\newblock{ \emph{Travaux mathematiques. Fasc. XVI}, 121--137,
Trav. Math., XVI, Univ. Luxemb., Luxembourg, 2005.}

\bibitem{Severa:2017}
P.~\v{S}evera,
\newblock{Letters to Alan Weinstein about Courant algebroids,}
\newblock{\texttt{ 	arXiv:1707.00265 [math.DG]} (2017).}

\bibitem{Sheng:2013}
 Y.~Sheng \& Z.~Liu,
 \newblock{Leibniz 2-algebras and twisted Courant algebroids,}
 \newblock{ \emph{Comm. Algebra.} \textbf{41} (2013), 1929-1953.}

\bibitem{Stienon:2008}
M.~ Sti\'{e}non \& P.~Xu,
\newblock{Modular classes of Loday algebroids,}
\newblock{\emph{C. R. Math. Acad. Sci. Paris} \textbf{346} (2008), 193--198.}

\bibitem{Vaisman:2005}
I.~ Vaisman,
\newblock{Transitive Courant algebroids,}
\newblock{\emph{Int. J. Math. Sci.} \textbf{11} (2005), 1737--1758.}

\bibitem{Vaintrob:1997}
A.Yu.~Va\u{\i}ntrob,
\newblock{Lie algebroids and homological vector fields,}
\newblock{ \emph{Russian Math. Surveys} \textbf{52} (1997), 428--429.}

\bibitem{Voronov:2001}
Th.Th.~Voronov,
\newblock {Graded manifolds and {D}rinfeld doubles for {L}ie bialgebroids,}
\newblock{ in: Quantization, Poisson Brackets and Beyond, volume \textbf{315} of
 \emph{ Contemp. Math}., pages 131--168. Amer. Math. Soc., Providence, RI, 2002.}

\bibitem{Voronov:2005}
Th.Th.~Voronov,
\newblock{Higher derived brackets and homotopy algebras,}
\newblock{ \emph{J. Pure Appl. Algebra} \textbf{202} (2005), no. 1-3, 133--153.}

\bibitem{Voronov:2012}
 Th.Th.~Voronov,
 \newblock{Q-manifolds and higher analogs of Lie algebroids,}
  \newblock{ XXIX Workshop on Geometric Methods in Physics, 191–202,
AIP Conf. Proc., 1307, Amer. Inst. Phys., Melville, NY, 2010. }

\bibitem{Xu:2014}
X.~Xu,
\newblock{Twisted Courant algebroids and coisotropic Cartan geometries,}
\newblock{\emph{J. Geom. Phys.} \textbf{82} (2014), 124--131.}

\end{small}
\end{thebibliography}
\end{document}